\documentclass[11pt]{article}
\makeatletter

\def\blfootnote{\gdef\@thefnmark{}\@footnotetext}

\makeatother
\usepackage[letterpaper,margin=1in]{geometry}

\usepackage[utf8]{inputenc} %
\usepackage[T1]{fontenc}    %
\usepackage{url}            %
\usepackage{booktabs}       %
\usepackage{amsfonts}       %
\usepackage{nicefrac}       %
\usepackage{microtype}      %
\usepackage{xcolor}         %

\usepackage[right, outer, inline, final]{showlabels} %
\usepackage{nicematrix}

\usepackage{tikz}
\usepackage{float}
\usetikzlibrary{shapes, 
                backgrounds, 
                patterns,
                calc,
                arrows,
                arrows.meta}
\usepackage{chemfig}

\catcode`\_=11
\definearrow5{-x>}{%
    \CF_arrowshiftnodes{#3}%
    \CF_expafter{\draw[}\CF_arrowcurrentstyle](\CF_arrowstartnode)--(\CF_arrowendnode)%
        coordinate[midway,shift=(\CF_ifempty{#4}{225}{#4}+\CF_arrowcurrentangle:\CF_ifempty{#5}{5pt}{#5})](line@start)%
        coordinate[midway,shift=(\CF_ifempty{#4}{45}{#4}+\CF_arrowcurrentangle:\CF_ifempty{#5}{5pt}{#5})](line@end)%
        coordinate[midway,shift=(\CF_ifempty{#4}{135}{#4}+\CF_arrowcurrentangle:\CF_ifempty{#5}{5pt}{#5})](line@start@i)%
        coordinate[midway,shift=(\CF_ifempty{#4}{315}{#4}+\CF_arrowcurrentangle:\CF_ifempty{#5}{5pt}{#5})](line@end@i);
    \draw(line@start)--(line@end);%
    \draw(line@start@i)--(line@end@i);%
    \CF_arrowdisplaylabel{#1}{0.5}+\CF_arrowstartnode{#2}{0.5}-\CF_arrowendnode
}
\catcode`\_=8

\usepackage{pict2e,picture,graphicx}
\usepackage[overload]{empheq}
\usepackage{microtype}
\usepackage{graphicx}
\usepackage{subcaption}
\usepackage{booktabs}
\usepackage{amsmath}
\usepackage{cases}
\usepackage{mathtools}
\usepackage{amsthm}
\usepackage{thm-restate}
\usepackage{enumitem}
\usepackage{xspace}
\usepackage{algorithm}
\usepackage[noend]{algpseudocode}


\usepackage{xcolor}
\usepackage{nicefrac, xfrac}
\usepackage{xparse}
\usepackage{hyperref}
\usepackage[capitalize, noabbrev]{cleveref}
\usepackage{wrapfig}
\usepackage{bbm}
\usepackage{yfonts}
\usepackage{nicefrac} 
\usepackage{multirow}
\usepackage{bm}
\usepackage{bbm}
\usepackage{subcaption}
\usepackage{MnSymbol}

\usepackage[T1]{fontenc}

\hypersetup{colorlinks,
    colorlinks = true,
    linkcolor=DarkBlue,
    citecolor=darkGreen,
    urlcolor=darkGreen,
    filecolor=DarkBlue
    linktocpage = true,
      }

\usepackage{comment}

\usepackage{colortbl}
\usepackage{cellspace}

\colorlet{darkgreen}{green!70!black}

\newcommand{\CFont}[1]{{\textup{\textsf{#1}}}\xspace}
\newcommand{\PPAD}{\CFont{PPAD}}

\newcommand{\NP}{\CFont{NP}}

\newcommand{\PTAS}{\CFont{PTAS}}
\newcommand{\QPTAS}{\CFont{QPTAS}}
\newcommand{\rng}{\mathsf{rg}}

\newcommand{\cT}{\mathcal{T}}
\newcommand{\cP}{\mathcal{P}}
\newcommand{\cF}{\mathcal{F}}

\newcommand{\cV}{\mathcal{V}}
\newcommand{\cH}{\mathcal{H}}
\newcommand{\cX}{\mathcal{X}}

\newcommand{\cB}{\mathcal{B}}

\newcommand{\poly}{\textnormal{poly}}

\def\ie{\textit{i.e.}\@\xspace}

\newcommand{\Reals}{\mathbb{R}}

\newcommand{\Naturals}{\mathbb{N}}

\usepackage{color}
\definecolor{mygreen}{rgb}{0.0, 0.5, 0.0}
\definecolor{myorange}{rgb}{0.55, 0.62, 1}

\theoremstyle{plain}

\newcommand{\child}{\mathsf{c}}

\theoremstyle{plain}
\newtheorem{theorem}{Theorem}[section]
\newtheorem{proposition}[theorem]{Proposition}

\newtheorem{lemma}[theorem]{Lemma}

\newtheorem{corollary}[theorem]{Corollary}
\theoremstyle{definition}
\newtheorem{definition}[theorem]{Definition}

\theoremstyle{remark}
\newtheorem{remark}[theorem]{Remark}

\definecolor{niceRed}{RGB}{190,38,38}
\definecolor{Red2}{RGB}{219, 50, 54}
\definecolor{mgreen}{RGB}{160, 200, 140}
\definecolor{blueGrotto}{RGB}{5,157,192}
\definecolor{limeGreen}{HTML}{81B622}
\definecolor{myellow}{rgb}{0.88,0.61,0.14}
\definecolor{darkGreen}{HTML}{2E8B57}
\definecolor{navyBlueP}{HTML}{03468F}
\definecolor{Sepia}{HTML}{7F462C}
\definecolor{red2}{HTML}{1F462C}
\definecolor{orange2}{HTML}{FF8000}
\definecolor{mgray}{HTML}{ABB3B8}
\definecolor{lgray}{HTML}{E5E8E9}
\definecolor{myPurple}{RGB}{175,0,124}
\definecolor{mypurple2}{rgb}{0.8,0.62,1}
\definecolor{royalBlue}{HTML}{057DCD}
\definecolor{mpink}{HTML}{FC6C85}
\definecolor{lblue}{RGB}{74,144,226}
\definecolor{peagreen}{RGB}{152,193,39}
\definecolor{typ_navy}{HTML}{001f3f}
\definecolor{typ_blue}{HTML}{0074d9}
\definecolor{typ_aqua}{HTML}{7fdbff}
\definecolor{typ_teal}{HTML}{39cccc}
\definecolor{typ_eastern}{HTML}{239dad}
\definecolor{typ_purple}{HTML}{b10dc9}
\definecolor{typ_fuchsia}{HTML}{f012be}
\definecolor{typ_maroon}{HTML}{85144b}
\definecolor{typ_red}{HTML}{ff4136}
\definecolor{typ_orange}{HTML}{ff851b}
\definecolor{typ_yellow}{HTML}{ffdc00}
\definecolor{typ_olive}{HTML}{3d9970}
\definecolor{typ_green}{HTML}{2ecc40}
\definecolor{typ_lime}{HTML}{01ff70}
\definecolor{newgreen}{HTML}{83c702}
\definecolor{newpurp}{RGB}{97,96,121}
\definecolor[named]{Purple}{cmyk}{0.55,1,0,0.15}
\definecolor[named]{DarkBlue}{cmyk}{1,0.58,0,0.21}

\usepackage[style=alphabetic,natbib=true,maxcitenames=2, maxbibnames=10,backend=bibtex]{biblatex}
\allowdisplaybreaks

\title{Quasi-Polynomial Simultaneous Approximation of Multiple Polynomials on Subsets of the Simplex}
\title{Quasi-Polynomial Optimization of Fixed-degree Polynomials Using Nets}
\title{A Unifying Framework for Quasi-Polynomial Optimization of Fixed-degree Polynomials}

\author{
    Martino Bernasconi$^\dagger$ \quad
    Matteo Castiglioni$^\ddagger$ \quad
    Andrea Celli$^\dagger$ \quad
    Gabriele Farina$^\ast$  \vspace{6mm}\\
    $^\dagger$\ Bocconi university\\
    $^\ddagger$\ Politecnico di Milano\\
    $^\ast$\ Massachusetts Institute of Technology\vspace{2mm}\\
    {\textcolor{black}{\small\texttt{\{martino.bernasconi,andrea.celli2\}@unibocconi.it}, \quad \texttt{matteo.castiglioni@polimi.it,}}}\\
    {\textcolor{black}{\small\texttt{gfarina@mit.edu}}}
}

\author{
\begin{tabular}{cc}
& \\
{Martino Bernasconi}\thanks{Martino Bernasconi and Andrea Celli were supported by an ERC grant (Project 101165466 — PLA-STEER).} & {Matteo Castiglioni}\thanks{Matteo Castiglioni was supported by the EU Horizon project ELIAS (European Lighthouse of AI for Sustainability, No. 101120237).}\\
\small{Bocconi University} & \small{Politecnico di Milano}\\
{\textcolor{black}{\small\texttt{martino.bernasconi@unibocconi.it}}} & %
{\textcolor{black}{\small\texttt{matteo.castiglioni@polimi.it}}}\\
& \\
{Andrea Celli}\footnotemark[1] & {Gabriele Farina}\thanks{Gabriele Farina was supported in part by the National Science Foundation award CCF-2443068, the Office of Naval Research grant N000142512296, and an AI2050 Early Career Fellowship.}\\
\small{Bocconi University} & \small{Massachusetts Institute of Technology}\\
{\textcolor{black}{\small\texttt{andrea.celli2@unibocconi.it}}} & %
{\textcolor{black}{\small\texttt{gfarina@mit.edu}}}
\end{tabular}
}
\date{}

\begin{document}

\maketitle
\begingroup
\renewcommand{\thefootnote}{}
\footnotetext{The authors thank Argyrios Deligkas, Themistoklis Melissourgos, and Pablo Parrilo for helpful comments and discussions.}
\addtocounter{footnote}{-1}
\endgroup
\begin{abstract}

    We study the simultaneous approximation of constant-degree polynomials over convex sets. For any family of $m$ degree-$d$ polynomials and any convex set $\mathcal{H} \subseteq \mathbb{R}_{\ge0}^n$, we construct an $\epsilon$-Cover of the joint value set $\{(f_1(x), \dots, f_m(x)) : x \in \mathcal{H}\}$ in the $\ell_\infty$-norm. This cover is of size $n^{O(\log(mn)/\epsilon^2)}$, provided the polynomials have constant range over the smallest $\ell_1$-ball inscribing $\mathcal{H}$.
    Our approach extends classical net-based sparsifications for linear functions (e.g., Lipton, Markakis, and Mehta [2003]) to arbitrary families of constant-degree polynomials over general convex sets.
    We use a two-step scheme: first, we construct a quasi-polynomial pre-cover of the family on the smallest $\ell_1$-ball containing $\cH$ by using a concentration argument and leveraging a connection between Bernstein approximation and multinomial distributions; we then compress the pre-cover to $\mathcal{H}$ by using a recursive degree reduction and feasibility programs anchored at points of the pre-cover.
    The existence of these covers immediately yields a unified framework for Quasi-Polynomial Time Approximation Schemes (QPTAS) across a wide range of a problems, including fixed-degree polynomial minimization over polyhedral sets, Constraint Satisfaction Problems (CSPs), Free Games, variational inequalities with polynomial operators (which implies guarantees for local Nash equilibria in polynomial games), and additive approximation for normalized densest $k$-subhypergraph on $O(1)$-uniform hypergraphs.
\end{abstract}

\tableofcontents

\section{Introduction}\label{sec:intro}

We provide a framework for constructing a quasi-polynomial-size net that simultaneously covers a family $\cF$ of fixed-degree polynomials on a polyhedral set (possibly with exponentially many vertices).%
\footnote{Our framework can also handle general sets that admit an efficient linear optimization oracle. For simplicity, we focus on polyhedral sets, and discuss in \Cref{rem:convex} how to extend our framework to such sets.}
Our framework handles all polynomials with constant range over the smallest $\ell_1$-ball inscribing $\cH$.

More precisely, let $\cH\subseteq \Reals_{\ge0}^n$  be a nonempty convex compact polytope, and let $\cF=\{f_1,\dots,f_m\}$ be a family of polynomials of constant degree $d$. Assume that every polynomial in $\cF$ has constant range over the smallest $\ell_1$-ball containing $\cH$. Then, for every $\varepsilon>0$, we construct a finite set $\cX\subseteq\cH$ of size $n^{O(\log(mn)/\varepsilon^2)}$, suppressing constants that depend only on the degree and the range bound, such that for every $x\in \cH$ there exists $x'\in \cX$ with $|f(x)-f(x')|\le \varepsilon$ for all $f\in \cF$.
We call such a set an $\varepsilon$-Cover of $\cF$ over $\cH$ (\Cref{def:cover}). Moreover, our proof is constructive and computes $\cX$ in quasi-polynomial time.

This result has a number of algorithmic implications for several optimization problems involving polynomials, including CSPs, Free Games \citep{aaronson2014multiple}, variational inequalities on polyhedral sets with a polynomial number of vertices (with notable implications for the computation of local Nash equilibria in nonconvex games \citep{daskalakis2021complexity, cai2024tractable}), the densest $k$-subhypergraph problem and more (see also \Cref{sec:applications}).

More broadly, our result unifies and extends several lines of research that had previously been developed in parallel. Although these works arose in different areas and used rather different techniques, their positive results (in particular, the existence of polynomial- or quasi-polynomial-time approximation schemes) can all be explained as direct consequences of our main theorem, for progressively richer instantiations. We next discuss these three lines of research in turn.

\begin{enumerate}[label={(\roman*)}, wide=0pt, align=left, labelwidth=.4cm]
    \item \label{item:intro1}
The first line of work concerns approximating the minimum of a single fixed-degree polynomial over the simplex.
Minimizing a polynomial over the probability simplex is a fundamental problem in computer science. Unfortunately, exact algorithms are already ruled out for degree-$2$ polynomials, as shown by the classical \NP-hardness result of \citet{motzkin1965maxima}. This led to a line of work asking whether one can nevertheless obtain polynomial-time approximation schemes \citep{bomze2002solving,nesterov2003random,de2006ptas}. This line culminated in a \PTAS for minimizing fixed-degree polynomials of bounded range over the simplex, with running time $n^{O(1/\varepsilon)}$. We remark that the small range condition is crucial for obtaining a \PTAS for additive approximation.
At the heart of this \PTAS, drawing a connection with positive forms on the simplex \citep{powers2001new}, is the fact that interpolation on a regular grid uniformly well-approximates fixed-degree polynomials over the simplex. Specifically, the algorithm enumerates all rational points in the simplex with denominator $N$, of which there are $n^{O(N)}$; it is shown that one such point approximates the optimum within an additive error of $\poly(1/N)$.

\item \label{item:intro2} A second line of research concerns sparsification of probability distributions by sampling, and in particular focuses on approximating the image of \emph{linear} functions over the simplex. Here, the main applications are driven by game theory. Perhaps the most famous result in this direction is the \QPTAS of \citet*{lipton2003playing}, which provides an algorithm that runs in quasi-polynomial time and finds a constant-approximation to Nash equilibria in normal-form games, building on preliminary results by \citet{althofer1994sparse} and \citet{lipton1994simple} for zero-sum games. 

These approximation results hinge on approximating finite distributions with sparse-support distributions. In particular, they rest on the fact that for all $x\in\Delta_n$, there is an $x'$ with support $O(\log(n)/\varepsilon^2)$ which additively approximates within $\varepsilon$ all payoffs in $\ell_\infty$-norm. For these results, the linearity of the payoffs is essential.
Related papers apply similar techniques to other settings such as community detection \citep{arora2012finding}, correlated equilibria \citep{babichenko2014simple, babichenko2017empirical}, signaling \citep{cheng2015mixture}, and Stackelberg equilibria \citep{gan2023robust}.
Having established the existence of near-optimal sparse-support distributions, algorithms then proceed by enumerating all possible such distributions, of which there are a quasi-polynomial number. While at first sight this procedure might appear identical to that of \Cref{item:intro1}, it is crucial to realize that the grid coarseness in the two applications must be significantly different (with Nash equilibria using a finer grid). This necessity shines through the complexity of the problems as well: while polynomial minimization admits a \PTAS, a \PTAS for approximating Nash equilibria cannot exist assuming the exponential time hypothesis for PPAD \citep{rubinstein2017settling}. Conceptually, in the Nash equilibrium application, the necessity of the finer grid stems from the fact that the chosen distribution must \emph{simultaneously} approximate \emph{multiple} (linear) polynomials---one for each action of the players. This is in contrast to \Cref{item:intro1}, for which one cares about approximating only \emph{one} function.

\item A third line of work moves beyond the full simplex and considers optimization of bounded-range polynomials over more general polyhedral domains. This extension is crucial for capturing sampling-based techniques in CSPs and Free Games, which admit low-degree polynomial algebrization over hypercubes rather than simplices.

At first glance, replacing the full simplex by a polyhedral subset might seem like a minor change. In reality, this change makes a drastic difference from a computational perspective.
Indeed, already minimizing a quadratic polynomial on a polyhedral subset of the $n$-simplex defined by $O(\text{poly}(n))$ linear constraints is expressive enough to capture problems for which \PTAS{}es are ruled out under standard complexity hypotheses. 
As a first example,
\citet{mangasarian1964two} showed that a Nash equilibrium in a two-player general-sum game can be expressed as a quadratic minimization problem over a polyhedral set with many vertices. The approach of \citep{de2006ptas} cannot lead to a \PTAS for this formulation, since the \QPTAS is known to be tight \citep{rubinstein2017settling}.
Another problem that can be cast in this form is CSPs, and in particular dense instances (\emph{i.e.}, instances with many constraints \citep{arora1995polynomial, frieze1996regularity}), in which sampling techniques are used to obtain \QPTAS{}es. Similarly, Free Games are a class of optimization problems arising from multi-prover interaction protocols with multiple independent provers \cite{aaronson2014multiple}. Both of these problems can be algebrized into low-degree polynomials over hypercubes. The surprising fact is that the intrinsic ``normalization'' of these problems (\emph{i.e.}, CSPs have many constraints, and Free Games are evaluated only on expectations) results in a small range of objective values over the relevant set.

Another example is the problem of additive approximation on a normalized version of the densest $k$-subgraph problem (NDkS) \citep{barman2018approximating,abboud2022improved}. The classic version of this problem (\emph{i.e.}, with multiplicative approximations) is a fundamental problem in theoretical computer science with connections to hardness of approximations \citep{feige2001dense, manurangsi2017almost} and other important problems \citep{pisinger2007quadratic, andersen2009finding, hajiaghayi2006prize}.
\citet{alon2013approximate} showed that NDkS can be written as a quadratic minimization problem on a polyhedral set with many vertices, and \citet{barman2015approximating} relied on this result to design a \QPTAS. Later, \citet{braverman2017eth} proved that, under the exponential time hypothesis, this result is tight.

\end{enumerate}

From a technical perspective, the key challenge is in showing how these covers can be constructed for arbitrary polyhedral domains. Unlike the full simplex, general domains may have complex geometries, precluding the kinds of trivial sampling and uniform-grid constructions that apply in the simplex. We address these hurdles by showing that covers can be adapted to the domain's geometry through a combination of two insights. The first is a recursive decomposition of the polynomials. The second is a compression step, implemented via linear programming, that maps a uniform grid on an ambient simplex to the domain of interest.

\subsection{Technical Overview}

Our approach is based on a sequence of refinements. First, we construct an $\varepsilon$-Cover for the entire simplex. Second, we extend the result to polyhedral subsets of the simplex. Finally, we extend it to arbitrary polyhedral sets by inscribing them into an $\ell_1$-ball. We now comment in more detail on these steps.

In \Cref{sec:oneprob}, we revisit the result of \citet*{de2006ptas}, which shows that uniform grids on the simplex are sufficient for approximating the minimum of a single polynomial, and provide a simple alternative probabilistic proof of this result.
Our probabilistic approach offers two key advantages. First, it yields a uniform approximation of the polynomial over the whole simplex. Second, it naturally combines with a union-bound argument to extend the result to a family of polynomials.

\begin{theorem}[Informal; full version in \Cref{cor:manyImageSimplex}]
    For each $\varepsilon>0$, there exists an algorithm which finds in $n^{O({\log(m)}/{\varepsilon^2})}$ time an $\varepsilon$-Cover for a family of $m$ polynomials of fixed degree and range over the simplex.
\end{theorem}

Our probabilistic approach builds on an observation by \citet*{de2015alternative}, who remark that the Bernstein approximation of a polynomial at a point $x\in\Delta_n$ can be interpreted as the expectation of the same polynomial evaluated at samples drawn from a (rescaled) multinomial distribution with parameter $x$. We embrace this observation as the starting point of our framework.

To show that the multinomial distribution induces an $\varepsilon$-Cover, we employ a concentration argument on top of the aforementioned Bernstein approximation bound.
This result exploits a fragile relationship between the directional derivatives of the polynomials along the tangent space of the simplex and their range. However, our argument cannot rely solely on the existence of good covers of the simplex (in the traditional sense) and translate those into a cover of the images via Lipschitzness. 
Indeed, even in the linear case, the Lipschitz constant of polynomial maps is not bounded under norms other than $\ell_1$, and Lipschitzness under the $\ell_1$-norm alone is insufficient, since no efficient covers of the simplex exist under this norm. We include a more in-depth discussion on this technical point in \Cref{sec:related}.

In \Cref{sec:many_subset}, we present our main technical result, \Cref{thm:manySubset}, showing the existence of quasi-polynomial covers of a family of polynomials over a \emph{subset} of the simplex.
We remark that considering a subset of the simplex, while seemingly artificial and somewhat of limited interest, is a technical intermediate result that already captures most of the technicalities of the general setting (which we will consider in \Cref{sec:extension}). Intuitively, this is motivated by the simple observation that ``every set in dimension $n$ can be inscribed in the simplex of dimension $n+1$''. The main result of \Cref{sec:many_subset} is the following.

\begin{theorem}[Informal; full version in \Cref{thm:manySubset}]
    For each $\varepsilon>0$, there exists an algorithm that finds in $n^{O({\log(mn)}/{\varepsilon^2})}$ time an $\varepsilon$-Cover for a given family of $m$ polynomials of fixed degree and range over a polyhedral subset of the simplex. 
\end{theorem}

Our construction proceeds in two steps. First, we design a ``pre-cover'' with quasi-polynomial size of the entire simplex. Next, we compress this pre-cover onto the target subset $\cH$ of the simplex. 
To achieve this, we iteratively decompose each polynomial in the family into lower-degree components through their derivatives. For each polynomial, we generate a tree of depth $d$ and branching factor $n$.
Intuitively, every node $f$ of degree $d$ has $n$ children $g_1,\ldots,g_n$ of degree $d-1$ satisfying $f(x)=\sum_{i}x_ig_i(x)$, where the $g_i$'s are related to the gradient of $f$.
Then, for each point $\tilde x$ in the ``pre-cover'', we define a feasibility LP that ``anchors'' the nodes of each tree at $\tilde x$. When feasible, this LP yields a point in $\cH$ that closely approximates the values of all polynomials at $\tilde{x}$. We point out that it is essential that the range is constant on the entire simplex, and not only on the polyhedral subset considered. We will discuss this point in more detail later in the paper.

Finally, in \Cref{sec:extension}, we extend our results from subsets of the simplex to arbitrary polyhedral sets. We show that the complexity of the problem is governed by a key parameter: the range of the polynomials over the smallest $\ell_1$-ball that inscribes the polyhedral set $\cH$. In particular, we define the diameter of the optimization domain $\cH$ as $\|\cH\|_1 = \max_{x \in \cH} \lVert x \rVert_1$. Then, we consider the range of the polynomial over the set $\{x\in \Reals_{\ge0}^n:\lVert x\rVert_1\le \|\cH\|_1\}$ (see \Cref{sec:extension} for a formal definition). Our framework results in a \textsf{QPTAS} whenever such a range is constant.

\begin{theorem}[Informal; full version in \Cref{thm:manySubsetArbitrary}]\label{th:manySubsetArbitraryInformal}
Consider a polyhedral set $\cH \subseteq \mathbb{R}_{\ge0}^n$ and a family of $m$ polynomials of fixed degree. Assume that the range of each polynomial $f \in \cF$ is constant over the smallest $\ell_1$-ball inscribing $\cH$. For any $\varepsilon > 0$, there exists an algorithm that finds in $n^{O(\log(nm)/\varepsilon^2)}$ time an $\varepsilon$-Cover for the family of polynomials over $\cH$.\end{theorem}

\subsection{Applications}
\label{sec:applications}

Our constructive algorithm implies that one can exhaustively search over the image of a vector of polynomials in quasi-polynomial time. This general fact has immediate consequences for a variety of search problems.

To start, it implies a \QPTAS for minimization/maximization problems $\min_{x \in \cH} f(x)$ for certain polyhedral sets $\cH$ and objective functions $f$.
Among them, we find Constraint Satisfaction Problems (CSPs) and Free Games, as described in \cref{sec:csp} and \cref{sec:free}.

Moving beyond applications involving a single polynomial, another general application domain for our construction is variational inequalities  (VIs) with polynomial operators over a polyhedral subset with polynomially many vertices, with implications for computing local Nash equilibria and local saddle points (\cref{sec:VI}).
To fix ideas, consider a variational inequality problem with constant-degree polynomial operators $\bm{f} : \Delta_n \to \Reals^n$, that is, the problem of finding $x \in \Delta_n$ such that
\[
    \bm{f}(x)^\top (x' - x) \le 0 \qquad \forall x' \in \Delta_n.
\]
To approximate the solution, we consider the family $\cF\coloneqq \{\bm{f}(x)^\top (e_1 - x), \dots, \bm{f}(x)^\top (e_n - x)\}$. We can then construct and enumerate the $\varepsilon$-Cover (which has size $n^{O(\log n/\varepsilon^2)}$), looking for any $x$ at which all polynomials in $\cF$ have value $\le \varepsilon$. As a noteworthy corollary, this result shows that one can find local Nash equilibria of games with constant-degree polynomial utilities in quasi-polynomial time, providing a partial positive counterpart to the results of \citep{anagnostides2026computational,bernasconi2026complexity}. %
We develop this result, together with some extensions, in \cref{sec:VI}.

Finally, our result clearly applies to combinatorial problems that can be algebraized into low-degree polynomials. As an application, we show in \cref{sec:hyper} how it yields additive approximations for the normalized densest $k$-subhypergraph problem on uniform hypergraphs. The densest $k$-subgraph problem is one of the most fundamental problems in theoretical computer science; however, its extension to hypergraphs \citep{chlamtac2018densest} is much less studied. This is probably because even for graphs, the problem cannot be approximated within any constant factor under multiplicative approximations \citep{alon2011inapproximability}. However, the additive version of the densest subgraph problem (under suitable normalization) is quasi-polynomial \citep{barman2018approximating, abboud2022improved}. We show that the tractability of the problem under additive approximations is extended to $O(1)$-uniform hypergraphs \citep{chlamtac2018densest}.
We reduce this problem to optimizing a degree-$d$ polynomial over a subset of the simplex, where $d$ is the size of a hyperedge.
We use a $\ell_2$-regularized generating polynomial of the hypergraph to show that an optimum of this polynomial on $\Delta_n\cap[0,1/k]^n$ can be efficiently converted into a solution to the densest $k$-subhypergraph problem. A crucial aspect of this reduction is bounding the range of the polynomials involved, which we achieve through Maclaurin’s inequality for elementary symmetric polynomials.
This yields a \QPTAS for constant-factor approximations on constant-degree uniform hypergraphs.

\subsection{Related Work}\label{sec:related}

An alternative way of producing $\varepsilon$-Covers of a family of polynomials over a set $\cH$ is to start from an $\varepsilon$-grid for the set $\cH$ under some $\ell_p$-norm and then use Lipschitz continuity of the polynomials in the family with respect to the same norm to obtain guarantees. 

For instance, \citet*{deligkas2020lipschitz} extended the classic result of \citet{lipton2003playing} beyond the linear utility setting to games with Lipschitz-continuous utilities (see also \citet{azrieli2013lipschitz}). 
In more detail, they design algorithms to find equilibria in games with payoffs that are $\lambda$-Lipschitz with respect to some $\ell_p$-norm (with $p\ge 2$), and obtain an algorithm running in time $n^{O(p\lambda^2/\varepsilon^2)}$. However, bounded constant-degree polynomials are generally not $O(1)$-Lipschitz in any norm other than the $\ell_1$ one (an example is given in \Cref{app:Lip}, \Cref{prop:Lip}, which holds even for degree $1$ polynomials). Hence, these results cannot apply to general bounded constant-degree polynomials on the simplex.

\citet*{deligkas2018approximating,deligkas2022approximating} motivated by the study of Existential Theory of the Reals \citep{schaefer2017fixed}, give \QPTAS{}es for special classes of polynomials, such as those with only a polylogarithmic number of terms and with a good enough Lipschitzness constant with respect to the $\ell_\infty$-norm.%
While it is known \citep{barman2018approximating} that there exist uniform grids of $\Delta_n$ in $\ell_p$-norm of size $n^{O(\min(p,\log(n))/\varepsilon^2)}$ for any $p\in[2,\infty]$, crucially the result does \emph{not} hold for $p=1$. Hence, the result does not apply to polynomials that are Lipschitz continuous only with respect to the $\ell_1$ norm, and therefore---in light again of the fact that general bounded low-degree polynomials can only be Lipschitz continuous with respect to the $\ell_1$-norm---this line of attack cannot recover a general result. %

\citet*[Theorem~5]{deligkas2022approximating} also generalize the sampling techniques of \citet{lipton1994simple, barman2018approximating} beyond the linear case and beyond the Lipschitzness assumption. Their result yields a quasi-polynomial-size cover for a family of polynomials over the simplex, using different techniques than those based on Bernstein interpolants \citep{de2006ptas} that we describe in \Cref{sec:oneprob}. More specifically, given $m$ polynomials each described as the sum of at most $t$ homogeneous terms of degree $\le d$ with coefficients in the range  $[-\alpha, \alpha]$, they give a $n^{ O(\poly(d,\alpha,t)\cdot \log(m)/\varepsilon^5)}$-sized $\varepsilon$-Cover. 
On the other hand, our result in \Cref{sec:oneprob} yields a cover of size $n^{ O(\poly(d^d)\cdot \log(m)/\varepsilon^2)}$, and thus a better on $\varepsilon$, but with a far worse dependence on the degree $d$. However, the assumption of bounded coefficients is not without loss of generality, as constant-range polynomials do not generally have constant-bounded coefficients.\footnote{For example, the polynomial
$
M(\sum_{i\in[n]} x_i)^2
$ 
has constant range and coefficients of magnitude $\Theta(M)$.} In general, constant-range polynomials can have coefficient magnitude bounded by $\alpha = d^{O(d)}$ (cf.~\Cref{sec:oneprob}); in that regime, the bound of \citet[Theorem~5]{deligkas2022approximating} recovers the $d^{O(d)}$ dependence in ours, while retaining a worst dependence on $\varepsilon$. In those cases where $\alpha \ll d^{O(d)}$, however, the result of \citet{deligkas2022approximating} leads to an interesting guarantee. We leave investigating those cases, including in combination with the general construction of \Cref{sec:many_subset}, as future work. %

Finally, \citet*[Theorem~5]{deligkas2022approximating} also extend their construction to domains that are products of simplices. These domains arise naturally in several of our applications of \Cref{sec:applications}, including the standard algebrization of CSPs, which instantiates a simplex over the alphabet for each variable.
In the case of product of $B$ simplices, $\Delta^B_q$, the algorithm of \citet{deligkas2022approximating} runs in time $q^{O(\poly(B ,1/\varepsilon,\log m))}$. This provides meaningful bounds only when $B$ is bounded. Unfortunately, this is often not an interesting regime: returning to the CSP example, a constant $B$ would yield a CSP with a constant number of variables, which can be solved by enumeration. In contrast, our results from \Cref{sec:extension} let us obtain (together with the compression technique developed in \Cref{sec:many_subset}, which is the main technical contribution of this work), a quasi-polynomial time algorithm for some important cases of polynomially large $B$ such as dense CSP (see other applications in \Cref{sec:applications_real}). %

\section{Notation and Preliminaries}

Given variables $x = (x_1, \dots, x_n)\in \Reals^n$, we denote with $\Reals_d[x_1,\ldots, x_n]$ the set of polynomials with real coefficients and degree at most $d$ over the variables $x_1,\ldots,x_n$. Since the paper mostly concerns polynomials over the simplex $\Delta_n$, we assume, without loss of generality, that all polynomials are homogeneous, that is, all monomials have the same degree $d$. The assumption is without loss of generality, as one can always multiply any monomial of degree $k<d$ by $(\sum_{i\in [n]}x_i)^{d-k}$, which is identically equal to $1$ on the simplex. %

The range of a polynomial $f \in \Reals_d[x_1,\dots,x_n]$ over a generic compact set $\cH \subseteq \Reals^n$ is denoted 
\[
    \rng_\cH(f) \coloneqq \max_{x\in \cH} f(x)- \min_{x \in \cH} f(x).
\]
Given the special role of the simplex in this paper, we let $\rng(f)$ be a shorthand for the range $\rng_{\Delta_n}(f)$ of $f$ over the simplex.
For a family of polynomials $\cF \subseteq \Reals_d[x_1,\ldots, x_n]$, $\rng_{\cH}(\cF)$ is defined as the maximum of the ranges of each function in the family over $\cH$, and similarly $\rng(\cF)$ is the range of the family over the simplex.

We use $\Naturals_d^n$ as the set of $n$-tuples of non-negative natural numbers that sum up to $d$. We will frequently use the following classic estimate $|\Naturals_d^n|=O(n^d)$.
Finally, for any $x\in\Reals^n$ and $\beta\in\Naturals_{d}^n$, we use the multinomial notation $x^\beta$ to denote $\prod_{i\in[n]}x_i^{\beta_i}$ and for $\beta\in\Naturals_{d}^n$ we define $\beta!=\prod_{i\in[n]}\beta_i!$. Finally, we use the symbol $\Reals_{\ge0}$ to denote the set of nonnegative reals.

In this paper, we denote by $\rho(d)$ a generic function that depends only on the degree $d$ of the polynomials. Since we study only the case where $d=O(1)$, we can think of $\rho$ as simply a (degree-dependent) constant.

\paragraph{The Bernstein basis.}
Sometimes, it will be helpful to consider the Bernstein basis of a polynomial, rather than the standard monomial basis. The Bernstein basis $\{\cB_\beta^d(x)\}_{\beta\in\Naturals_d^n}$ of $\Reals_d[x_1,\ldots,x_n]$ is given by monomials $\cB_\beta^d(x)=\frac{d!}{\beta!}x^\beta$.
The coefficients with respect to the Bernstein basis (particularly, their range) are tightly connected to the range of the polynomial on the simplex, as we recall next. 

\begin{lemma}[{\cite[Theorem 2.2]{de2006ptas}}]\label{lem:boundBerncoeff}
    Let $f\in\Reals_d[x]$ be an $n$-variate polynomial and let $\{c_\beta\}_{\beta\in\Naturals_d^n}$ be the coefficients of $f$ with respect to the Bernstein basis. Then,
    \[
    \max_{\beta\in\Naturals^n_d}c_\beta-\min_{\beta\in\Naturals^n_d}c_\beta\le\rho(d)\cdot\rng(f),
    \]
    where $\rho(d)$ is a function that depends only on the degree $d$.
\end{lemma}
\begin{remark}
    The specific representation of the input polynomials is inconsequential. Even under oracle access, one can recover the representation of a polynomial in the Bernstein basis by evaluating the polynomial on the $d$-uniform grid and solving a linear program of size $O(n^d)$ to infer the corresponding coefficients.
\end{remark}

\paragraph{Cover of a family.}
One of our main results is an algorithm for simultaneously approximating the image of a family of polynomials $\cF$ over a compact set $\cH$. To formalize this result, a staple of our technical construction will be the notion of \emph{$\varepsilon$-Cover of $\cF$ over $\cH$}, which is a covering of $\cH$ under the pseudometric $d_\cF(x,y)=\sup_{f\in\cF}|f(x)-f(y)|$ induced by the class $\cF$, as formalized in \Cref{def:cover}.

\begin{definition}[$\varepsilon$-Cover of $\cF$ over $\cH$]\label{def:cover}
Given an $\varepsilon>0$ and a family of functions $\cF$, we say that a set $\cX_\cH^\cF$ is an $\varepsilon$-Cover of $\cF$ over $\cH\subseteq \Reals^n$ if for all $x \in \cH$ there exists a $x' \in \cX\subseteq \cH$ such that $|f(x)-f(x')|\le \varepsilon$ for all $f\in\cF$.
\end{definition}

The main technical result we will prove is that for a family of fixed degree polynomials $\cF$ and any constant $\varepsilon$, there always exists an $\varepsilon$-Cover of quasi-polynomial size of $\cF$ over polyhedral subsets of the simplex $\cH\subseteq\Delta_n$. Moreover, the proof will be constructive. Similar questions for the case $d=1$ were considered by \citet*{alon2013approximate}.

\section{Probabilistic Proof of the \PTAS for Optimizing Fixed-Degree Polynomials over the Simplex} \label{sec:oneprob}

\citet*{de2006ptas} proved 
the existence of a polynomial-time algorithm for finding constant additive approximations of the minimum value of a polynomial of constant degree over the simplex.
Their result heavily relies on the \emph{Bernstein approximation} of polynomials. 

\begin{definition}
The Bernstein approximation of order $N$ of an $n$-variate polynomial $f\in \Reals_d[x]$ is the polynomial
\begin{align}\label{eq:Bernstein}
B_N[f](x)\coloneqq\sum_{\beta\in \Naturals^n_N} f\left(\frac{\beta}{N}\right)\frac{N!}{\beta!}x^\beta.
\end{align}
\end{definition}

As a stepping stone, in this section we leverage the probabilistic properties of the Bernstein approximation and a simple concentration argument to prove results that generalize those of \citet{de2006ptas}. 
The probabilistic nature of the result will help us generalize the construction in two directions: 1) provide a \emph{uniform} approximation of the polynomial; and 2) extend the result to a \emph{family} of polynomials. The proof of the next statement is deferred to \Cref{app:prob}.

\begin{restatable}{theorem}{theoremEasy}\label{th:probabilisticPTAS}
    Let $f\in\Reals_d[x]$. For any $x\in\Delta_n$ and $\varepsilon>0$, let $Z=\sum_{k=1}^N X_k/N$, where $X_k\stackrel{i.i.d.}{\sim} \mathrm{Categorical}(x)$ and $N=\Omega(\frac{1}{\varepsilon}\rho(d)\cdot\rng(f))$. Then the following holds:
    \[
    \mathbb{P}(|f(Z)-f(x)|\ge \varepsilon)\le 2\exp\left(-\frac{N\varepsilon^2}{\rho(d)\cdot\rng(f)^2}\right).
    \]
\end{restatable}

A direct consequence of \Cref{th:probabilisticPTAS} is that for every $x\in\Delta_n$ (and not only for the minimizer of $f$ over $\Delta_n$), whenever $N=\Omega({\rng(f)^2\rho(d)}/{\varepsilon^2})$, there exists a realization of $Z$ such that $|f(Z)-f(x)|\le \varepsilon$. Thus, the set of all $n^{O({\rng(f)^2 \rho(d)}/{\varepsilon^2})}$ possible realizations of $Z$ defines an $\varepsilon$-Cover of $f$.

\begin{remark}
   The algorithm resulting from our probabilistic proof of the $\PTAS$ (\cref{th:probabilisticPTAS}) runs slower than the one of \cite{de2006ptas} in terms of $\varepsilon$ (inversely quadratic in our case, rather than inversely linear in \cite{de2006ptas}). However, the probabilistic nature of our proof makes the result stronger, as it enables us to extend it to the joint optimization of multiple polynomials in the next subsection. To our knowledge, all results that approximate many functions exhibit a $1/\varepsilon^2$ dependence \citep{lipton2003playing, barman2018approximating, aaronson2014multiple}. It is an interesting open question whether our technique can be made to work in $n^{O(1/\varepsilon)}$ rather than $n^{O(1/\varepsilon^2)}$.
   This gap was also observed by \citet{aaronson2014multiple}.
\end{remark}

\subsection{Quasi-Polynomial Covers: Family of Polynomials over the Simplex}\label{sec:many_entire}

We now point out that the probabilistic proof of \Cref{th:probabilisticPTAS} (which works for a single polynomial) can be combined with a union-bound argument to obtain an $\varepsilon$-Cover for a family of polynomials on the simplex. Although the argument is simple, it already leads to interesting applications, including to Nash equilibrium approximation, as we discuss in \cref{sec:VI}. Extending the result to polyhedral subsets is significantly more challenging and is deferred to \cref{sec:many_subset}.

\begin{theorem}\label{cor:many}
    Let $\cF$ be a family of $m$ polynomials of degree $d$. Moreover, for any $x\in\Delta_n$ and any $\varepsilon>0$, let $Z=\sum_{k=1}^N X_k/N$, where $X_k\stackrel{i.i.d.}{\sim} \mathrm{Categorical}(x)$ and $N=\Omega(\frac{1}{\varepsilon}\rho(d)\cdot\rng(\cF))$. Then,
    \[
    \mathbb{P}(|f(Z)-f(x)|\le \varepsilon\ \forall f\in\cF)\ge 1-2m\exp\left(-\frac{N\varepsilon^2}{\rho(d) \cdot\rng(\cF)^2}\right).
    \]
\end{theorem}

As a corollary of \Cref{cor:many}, we get that for all $x\in\Delta_n$ whenever $N=\Omega(\frac{\rng(\cF)^2\rho(d)\log(m)}{\varepsilon^2})$, there exists at least one realization of the variable $Z$ such that $|{f}(Z)-{f}(x)|\le\varepsilon$ for all $f\in\cF$. 

Thus, we can iterate over all possible realizations of $Z$ and build a set of size $|\Naturals_{N}^n|=n^{O(N)}$ %
with $N=\Theta(\frac{\rng(\cF)^2\rho(d)\log(m)}{\varepsilon^2})$ that is an $\varepsilon$-Cover of $\cF$ over $\Delta_n$.
Formally:
\begin{theorem}\label{cor:manyImageSimplex}
    Given a family of $m$ degree-$d$ polynomials $\cF$ and any $\varepsilon>0$, there exists an algorithm which finds in $n^{O\big({\rng(\cF)^2\rho(d)\log(m)}/{\varepsilon^2}\big)}$ time an $\varepsilon$-Cover $\cX_{\Delta_n}^\cF$ of $\cF$ over the simplex. Moreover, $|\cX_{\Delta_n}^\cF|=n^{O\big({\rng(\cF)^2\rho(d)\log(m)}/{\varepsilon^2}\big)}$.
\end{theorem}

Notice that an exponential dependence on $\log(m)$ is needed. Indeed, this result can be used to prove the $n^{O(\log(n))}$ algorithm for constant approximations for $n$ actions bimatrix games of \citet{lipton2003playing} (see \Cref{sec:VI} for details), which is tight assuming the exponential time hypothesis for \PPAD \citep{rubinstein2017settling}.

\begin{remark}\label{rem:polytope}
It is easy to see that all our results will hold for generic convex hulls of $n$ points $\cV:=\textnormal{co}(v_1,\ldots, v_n)\subset\Reals^k$, without adding a dependence on the dimension $k$. Formally, for any polynomial $f$ and any $x\in\Delta_n$, we can define the polynomial $f'(x)=f(Vx):\Delta_n\to\Reals$ (where $V\in\Reals^{k\times n}$ is the matrix whose columns are $v_1,\ldots, v_n$). It is evident that $\rng(f')=\rng_{\cV}(f)$, (\emph{i.e.}, the range of $f'$ on $\Delta_n$ is the range of $f$ over $\cV$), and that given an $\varepsilon$-Cover of $f'$ over $\Delta_n$, one can find a $\varepsilon$-Cover of $f$ over $\cV$, simply by mapping each point through $V$.
\end{remark}

\section{Quasi-Polynomial Covers: Family of Polynomials over a Subset of the Simplex}\label{sec:many_subset}

En route to the result for arbitrary sets of \cref{sec:extension}, in this section we show how one can construct an $\varepsilon$-Cover of a family of polynomials over a \emph{polyhedral subset} of the simplex $\cH=\{x\in\Delta_n:Ax\le b\}$ (possibly with exponentially many vertices). As discussed in \Cref{sec:intro}, this extension requires new ideas, as 
even from a complexity standpoint there are major differences between the constrained and the unconstrained simplex case. For perspective, minimizing a polynomial on the simplex is \NP-hard, but admits a \PTAS for constant-degree polynomials. In contrast, approximating the minimum of a quadratic form on a \emph{subset} of the simplex requires at least quasi-polynomial time \citep{braverman2017eth}, assuming established complexity theoretic assumptions such as the Exponential Time Hypothesis.

    This might seem counterintuitive; indeed, one might be under the impression that it would be sufficient to take a cover of the family over $\Delta_n$, and intersect it with $\cH$. However, this simple guess is bound to fail. Intuitively, the cover produced on $\Delta_n$ is {extremely} sparse. In high dimensions, it might well be that none of the grid points of $\Delta_n$ fall in $\cH$. This simple construction must also fail due to the complexity considerations mentioned above. If one could simply minimize a bounded-range fixed-degree polynomial on a subset $\cH$ of the simplex by constructing a uniform grid on the simplex and taking the smallest value among the grid points that fall in $\cH$, one would obtain a \PTAS rather than a \QPTAS, and thus break the lower bound of \citet{braverman2017eth}.

    Crucially, the results in this section are much stronger than what can be obtained by reparameterizing $\cH$ as a convex hull of its vertices (\cref{rem:polytope}). In particular, on polyhedral subsets $\cH \subseteq \Delta_n$ with exponentially many vertices, the construction of this section yields an $\varepsilon$-Cover whose size scales as $n^{\log n}$ with respect to the number of variables $n$ (assuming constant range of the polynomial).

Formally, in this section, we prove the following result.
\begin{theorem}\label{thm:manySubset}
    Let $\cF$ be a family of degree-$d$ polynomials, and let $\cH\subseteq \Delta_n$ be a polytope defined by $k$ inequalities. There exists an algorithm that runs in time $\poly(k)\cdot n^{O(\rng(\cF)^{2}\rho(d)\log(n|\cF|)/\varepsilon^2)}$ that 
    returns a discrete set $\cX\subseteq \cH$ of size $n^{O(\rng(\cF)^{2}\rho(d)\log(n|\cF|)/\varepsilon^2)}$ which is an $\varepsilon$-Cover of $\cF$ over $\cH$.
\end{theorem}

We will extend the result to arbitrary convex subsets of the simplex (as long as they admit an efficient linear optimization oracle) in \Cref{rem:convex}.
As a direct corollary, by enumerating over the $\varepsilon$-Cover produced by \Cref{thm:manySubset}, we obtain the following.
\begin{corollary}\label{cor:min_subset}
    Let $f$ be a degree-$d$ polynomial, and $\cH \subseteq \Delta_n$ be a polytope defined by $k$ inequalities. There exists an algorithm that runs in time $\mathrm{poly}(k)\cdot n^{O(\rng(f)^{2}\rho(d)\log(n)/\varepsilon^2)}$ and return an $x\in \cH$ such that $f(x)\le \min_{x\in\cH} f(x)+\varepsilon$.
\end{corollary}

\subsection{Proof of \Cref{thm:manySubset}}

In this section, we will provide a proof of our main result, namely, \Cref{thm:manySubset}. The proof's idea is to compress a ``pre-cover'' for (an enlarged family of) $\cF$ over $\Delta_n$ to a cover over $\cH$ by using two technical tools: 1) feasibility linear programs to generate the cover over $\cH$, and 2) recursively decomposing the polynomials of degree $k$ into components of lower degree $k-1$.

To form the basis of the induction, we start from the case of degree-one polynomials. %

\subsubsection{Warm Up: Linear Functions ($d=1$)}

Consider a set $\cH=\{x\in\Delta_n: Ax\le b\}$ and a family $\cF$ of $m$ degree-$1$ polynomials that we can represent with a $m\times n$ matrix $B$. In \Cref{cor:manyImageSimplex}, we showed how to build a $\varepsilon$-Cover $\cX^{\cF}_{\Delta_n}$ of $\cF$ over the entire simplex $\Delta_n$.

We are going to compress the ``pre-cover'' $\cX^{\cF}_{\Delta_n}$ over $\Delta_n$ to a cover $\cX^{\cF}_{\cH}$ over $\cH$ by solving a linear feasibility program for each element in the cover. %
Specifically, for each $\tilde x\in \cX^{\cF}_{\Delta_n}$, we can consider the feasibility LP $\cP(\tilde x)$ given by 
\Cref{eq:feas1} of \Cref{alg:d=1}.
If $\cP(\tilde x)$ is feasible, we add its solution to the cover $\cX^{\cF}_{\cH}$.

We now argue that the set $\cX_\cH^{\cF}$ returned by \Cref{alg:d=1} is indeed an $\varepsilon$-Cover of $\cF$ over $\cH$ of quasi-polynomial size, and the algorithm runs in quasi-polynomial time.
First, we prove that the set built by the algorithm is indeed an $\varepsilon$-Cover of $\cF$ over $\cH$.
Consider any $x\in\cH$. Since $\cX_{\Delta_n}^{\cF}$ is an $\varepsilon/2$-Cover on the whole simplex, at least one point $\tilde x\in\cX_{\Delta_n}^{\cF}$ is such that $|f(x)-{f}(\tilde x)|\le\varepsilon/2$ for all $f\in \cF$.
Moreover, consider the feasibility LP $\cP(\tilde x)$ for that specific point. This LP is feasible, since $x$ itself is a solution. Thus, $\cX_{\cH}^{\cF}$ contains a solution $\hat x$ (possibly different from $x$)  of $\cP(\tilde x)$.
Then, for all $f\in\cF$, we can write the chain of inequalities
    \begin{align*}
        |f(x)-f(\hat x)|&\le |f(x)-f(\tilde x)|+|f(\tilde x)-f(\hat x)|     \le\frac\varepsilon2+\frac\varepsilon2,
    \end{align*}
    where the last inequality holds by definition of $\tilde x$ and by feasibility of $\hat x$ with respect to~$\cP(\tilde x)$.

    Now, for the second part of the statement, it is clear that solving $\cP(\tilde x)$ (or certifying its infeasibility) takes at most polynomial time for each $\tilde x$, and we have to consider at most $|\cX_{\Delta_n}^{\cF}|$ many such LPs. Thus, \cref{alg:d=1} runs in time $\poly(k)\cdot n^{O(\rng(\cF)^2\cdot\log(m)/\varepsilon^2)}$, and the size of the produced cover is at most $n^{O(\rng(\cF)^2\cdot\log(m)/\varepsilon^2)}$.

\begin{algorithm}[!ht]
    \caption{Construction of an $\varepsilon$-Cover for a family of \emph{linear} functions}\label{alg:d=1}
    \begin{algorithmic}[1] %
        \Require{A family of degree $1$ polynomials $\cF$ and a polyhedral $\cH\subseteq \Delta_n$}
        \State{Build an $\frac\varepsilon2$-Cover $\cX_{\Delta_n}^{\cF}$ of $\cF$ over $\Delta_n$}
        \State Define the $m\times n$ matrix $B$, where each row is one of the polynomials of $\cF$
        \State Initialize $\cX_{\cH}^{\cF}\gets\{\emptyset\}$
        \For{$\tilde x\in\cX_{\Delta_n}^{\cF}$}
        \State Let \begin{align}\label{eq:feas1}
        \cP(\tilde x):\quad\text{find $x\in\Delta_n$ \quad s.t.\quad} 
        \begin{cases}
        Ax\le b\\
        \|B\tilde x-Bx\|_\infty\le \frac{\varepsilon}{2}
        \end{cases}
        \end{align}
            \If{$\cP(\tilde x)$ is feasible}
                \State Add a solution of $\cP(\tilde x)$ to $\cX_{\cH}^{\cF}$
            \EndIf
        \EndFor
    \end{algorithmic}
\end{algorithm}

\subsubsection{The Case $d>1$}

To extend this approach to degrees greater than one, we need the following technical lemma, which shows how any polynomial $f$ of degree $d$ on the simplex can be written as a convex combination of $n$ polynomials (related to the derivatives of $f$) of degree $d-1$. Importantly, such polynomials have a range that is only a $d$-dependent multiplicative factor larger than the range of the original polynomial.

\begin{lemma}\label{lem:decomposition}
Let $f(x)=\sum_{\beta\in\Naturals_{d}^n} c_\beta \cB^d_{\beta}(x)$ be a degree $d$ polynomial expressed in the Bernstein basis. Then, for all $x\in\Delta_n$, we can write $f(x)=\sum_{i\in[n]} x_i g_i(x)$ where
\[
    g_i(x)=f(0)+\sum_{\alpha\in\Naturals_{d-1}^n}c_{\alpha+e_i} \cB_{\alpha}^{d-1}(x).
\]
Moreover, $\rng(g_i)\le \rho(d) \cdot \rng(f)$ for all $i\in[n]$.
\end{lemma}

\begin{proof}
    We write $f(x)=f(0)+\int_{0}^1\frac{d}{dt} f(tx)dt$. By the chain rule we have $\frac{d}{dt}f(tx)=\nabla f(tx)^\top x$, from which we obtain:
    \begin{align}\label{eq:tmp1}
    f(x)-f(0)=\int_0^1 \sum_{i\in[n]}x_i\partial_{x_i}f(tx) dt=\sum_{i\in[n]}x_i\int_0^1\partial_{x_i}f(tx)dt
    \end{align}
    Recall the definition of the Bernstein basis: $\cB_\alpha^d(x)=\frac{d!}{\alpha!}x^\alpha$ for all $\alpha\in\Naturals^n_{d}$. Now, it is clear that
    \[
    \partial_{x_i} \cB_\alpha^d(x)=\mathbb{I}(\alpha_i>0)\frac{d!}{\alpha!}\alpha_i x^{\alpha-e_i}=\mathbb{I}(\alpha_i>0)d\frac{(d-1)!}{(\alpha-e_i)!}x^{\alpha-e_i}=d \cdot\mathbb{I}(\alpha_i>0)\cB^{d-1}_{\alpha-e_i}(x).
    \]
    Thus,
    \[
    \partial_{x_i} f(tx)=d\sum_{\alpha\in\Naturals_{d}^n} \mathbb{I}(\alpha_i>0)c_\alpha \cB_{\alpha-e_i}^{d-1}(tx)=d\sum_{\alpha\in\Naturals_{d}^n} \mathbb{I}(\alpha_i>0)c_\alpha \cB_{\alpha-e_i}^{d-1}(x) t^{d-1},
    \]
    from which we obtain
    \[
    \int_0^1\partial_{x_i} f(tx)dt=d\sum_{\alpha\in\Naturals_{d}^n} \mathbb{I}(\alpha_i>0)c_\alpha \cB_{\alpha-e_i}^{d-1}(x)  \int_0^1 t^{d-1} dt = \sum_{\alpha\in\Naturals_{d}^n} \mathbb{I}(\alpha_i>0)c_\alpha \cB_{\alpha-e_i}^{d-1}(x)=\sum_{\alpha\in\Naturals^n_{d-1}} c_{\alpha+e_i} \cB_\alpha^{d-1}(x),
    \]
    which concludes the first part of the proof. Indeed, by using the last identity in \Cref{eq:tmp1} we obtain:
    \begin{align*}
    f(x)&=f(0)+\sum_{i\in[n]}x_i \cdot \sum_{\alpha\in\Naturals^n_{d-1}} c_{\alpha+e_i} \cB_\alpha^{d-1}(x)=\sum_{i\in[n]} x_i\left(f(0)+\sum_{\alpha\in\Naturals^n_{d-1}} c_{\alpha+e_i} \cB_\alpha^{d-1}(x)\right).
    \end{align*}

    Now, for the second part of the lemma, we observe that 
    \[
    g_i(x)-f(0)=\sum_{\alpha\in\Naturals^n_{d-1}} c_{\alpha+e_i} \cB_\alpha^{d-1}(x)
    \]
    and, thanks to the multinomial theorem, we also have that for all $x\in\Delta_n$:
    \[
    1=\left(\sum_{i=1}^nx_i\right)^{d-1}=\sum_{\beta\in\Naturals_{d-1}^n}\frac{(d-1)!}{\beta!}x^{\beta}=\sum_{\beta\in\Naturals_{d-1}^n} \cB^{d-1}_\alpha(x).
    \]
    Thus, $g_i(x)-f(0)$ is a convex combination of (some of) the coefficients $\{c_{\alpha}\}_{\alpha\in\Naturals_{d}^n}$ and thus
    \[
    g_i(x)-f(0)\le \max_{\alpha\in\Naturals_{d}^n}c_\alpha
    \quad\text{and}\quad
    g_i(x)-f(0)\ge\min_{\alpha\in\Naturals_{d}^n}c_\alpha.
    \]
    Hence, we can conclude that
    \[
    \rng(g_i)=\sup_{x\in\Delta_n}g_i(x)-\inf_{x\in\Delta_n}g_i(x)\le \max_{\alpha\in\Naturals_{d}^n}c_\alpha-\min_{\alpha\in\Naturals_{d}^n}c_\alpha\le \rho(d)\cdot\rng(f).
    \]
    where the last inequality is from \citet[Theorem 3]{de2015alternative}.
\end{proof}

Now, by using \Cref{lem:decomposition}, we can extend the algorithm from degree $1$ to degree $d$ polynomials by recursively decomposing each polynomial.
\begin{figure}
    \centering
    \scalebox{0.8}{\tikzset{every picture/.style={line width=0.75pt}} %

\begin{tikzpicture}[
    x=0.75pt,y=0.75pt,yscale=-1,xscale=1,
    every node/.style={
        font=\large,
        fill=white,
        inner sep=4pt,
        draw=none
    }
]

\node (root) at (370,310) {$4x_{1}^{3} +2x_{1}^{2} x_{2} +x_{1} x_{2}^{2}$};

\node (L1) at (240,410) {$4x_{1}^{2} +\frac{4}{3} x_{1} x_{2} +\frac{1}{3} x_{2}^{2}$};
\node (R1) at (500,410) {$\frac{2}{3} x_{1}^{2} +\frac{2}{3} x_{1} x_{2}$};

\node (L2a) at (170,510) {$4x_{1} +\frac{4}{6} x_{2}$};
\node (L2b) at (310,510) {$\frac{4}{6} x_{1} +\frac{1}{3} x_{2}$};
\node (R2a) at (430,510) {$\frac{2}{3} x_{1} +\frac{1}{3} x_{2}$};
\node (R2b) at (570,510) [inner sep=0pt] {$\frac{2}{6} x_{1}$};

\draw (root) -- (L1) node [midway] {$x_1$};
\draw (root) -- (R1) node [midway] {$x_2$};
\draw (L1)   -- (L2a) node [midway] {$x_1$};
\draw (L1)   -- (L2b) node [midway] {$x_2$};
\draw (R1)   -- (R2a) node [midway] {$x_1$};
\draw (R1)   -- (R2b) node [midway] {$x_2$};

\end{tikzpicture}}
    \caption{Example of the tree decomposition of the polynomial $f(x)=4x_1^3+2x_1^2x_2+x_1x_2^2\in\Reals_3[x_1,x_2]$. For instance, looking at the right subtree, for which we have $\frac23x_1^2+\frac23x_1x_2=x_1\cdot\left(\frac23x_1+\frac13x_2\right)+x_2\cdot \frac26x_1$.}
    \label{fig:tree}
\end{figure}
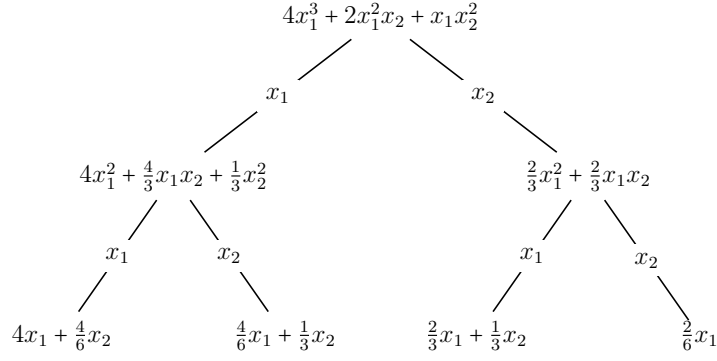
Formally, for each polynomial $f\in\cF$, we build a tree $\cT(f)$ of depth $d$ with the following properties (see \Cref{fig:tree} for an example).
\begin{enumerate}[label=$\arabic*)$]
    \item The root node is the polynomial $f$;
    \item For each node $N$ at layer $k$ (thus of degree $d-k+1$), we apply the decomposition of \Cref{lem:decomposition} to generate a family of $n$ polynomials of degree $d-k$, which constitute the children of $N$;
    \item At layer $k$, we only have polynomials of degree $d-k+1$;
    \item The tree has a branching factor of $n$ and depth $d$;
    \item For each node/polynomial $g$ at layer $k\le d-1$, i.e., of degree at least two, we denote with $\child[g]$ the set of children of $g$, and in particular, for $i\in[n]$, $\child_i[g]$ is the $i$-th child of $g$; 
    \item The leaves of the tree are degree 1 polynomials;
    \item For each leaf $g$ of the tree, which corresponds to a degree-$1$ polynomial $g(x)=b^\top x$, we define, with a slight abuse of notation, $\child_i[g]=b_i$;
    \item By construction, it holds that $\sum_{i\in[n]} x_i \cdot \child_i[g](x)=g(x)$ for all $g\in\cT(f)$;
    \item There are at most $n^d$ nodes.
\end{enumerate}

\begin{algorithm}[t]
    \caption{Construction of an $\varepsilon$-Cover for a family of degree-$d$ polynomials}\label{alg:d}
    \begin{algorithmic}[1] %
        \Require{A family of degree $d$ polynomials $\cF$}
        \State Set $\delta=\frac{\varepsilon}{4d}$

        \For{$f\in\cF$}\Comment{Recursively decompose $f$}
            \State Create recursively $\cT(f)$ by repeatedly applying \Cref{lem:decomposition}
        \EndFor

        \State Define the family $\widehat{\cF}=\{g\in\cT(f):f\in\cF\}$

        \State{Build an $\delta$-Cover $\cX_{\Delta_n}^{\widehat{\cF}}$ of $\widehat{\cF}$ over $\Delta_n$}

        \State Initialize $\cX_{\cH}^{\cF}\gets\{\emptyset\}$
        \For{$\tilde x\in\cX_{\Delta_n}^{\widehat{\cF}}$}
        \State Let \begin{align}\label{eq:feas2}
        \cP(\tilde x):\quad\text{find $x\in\Delta_n$ \quad s.t.\quad} 
        \begin{cases}
        Ax\le b\\
        \left|\sum_{i\in[n]}x_i\cdot \child_i[\hat f](\tilde x)-\hat f(\tilde x)\right|\le 2\delta\quad\forall \hat f\in\widehat{\cF}
        \end{cases}
        \end{align}
            \If{$\cP(\tilde x)$ is feasible}
                \State Add a solution of $\cP(\tilde x)$ to $\cX_{\cH}^{\cF}$
            \EndIf
        \EndFor
    \end{algorithmic}
\end{algorithm}

For a family $\cF$ of polynomials of degree $d$, we define a family $\widehat\cF$ of $O(|\cF|\cdot n^d)$ polynomials (of degree at most $d$) as the union of all nodes of the $|\cF|$ trees:
\[
\widehat \cF = \left\{ g\in \cT(f): f\in\cF\right\}.
\]

Next, we consider the $\delta$-Cover $\cX_{\Delta_n}^{\widehat\cF}$ of this family over the entire simplex (for some $\delta>0$ not yet defined). This serves the role as the simplex ``pre-cover'' of the previous section, and thus we would like to compress it to a cover over $\cH$.
In order to do this, we proceed according to the steps outlined in \Cref{alg:d}. For each $\tilde x\in\cX_{\Delta_n}^{\widehat\cF}$ we are going to instantiate the following feasibility LP:
\[
\cP(\tilde x): \quad \text{find $x\in\Delta_n$\quad s.t.\quad}
\begin{cases}
Ax\le b\\
\left|\sum_{i\in[n]} x_i\cdot \child_i[\hat f](\tilde x) - \hat f(\tilde x)  \right|\le2\delta\quad \forall \hat f\in \widehat\cF,
\end{cases}
\]
and add its solution to $\cX_\cH^{\cF}$ when $\cP(\tilde x)$ is feasible. 

For all $x\in\cH$ we let $\tilde x\in\cX_{\Delta_n}^{\widehat\cF}$ be such that $|\hat f(x)-\hat f(\tilde x)|\le \delta$ for all $\hat f$, which is guaranteed to exists by definition of $\delta$-Cover. As before, we remark that the corresponding $\cP(\tilde x)$ must be feasible, as $x$ itself is a solution to $\cP(\tilde x)$. Indeed $x\in \cH$, and moreover
\begin{align*}
    \left|\sum\nolimits_{i\in[n]}x_i\cdot \child_i[\hat f](\tilde x)-\hat f(\tilde x)\right|&\le \left|\sum\nolimits_{i\in[n]}x_i\cdot \child_i[\hat f](\tilde x)-\hat f( x)\right|+\left|\hat f( x)-\hat f(\tilde x)\right|\\
    &\le\left|\sum\nolimits_{i\in[n]}x_i \cdot (\child_i[\hat f](\tilde x)-\child_i[\hat f]( x))\right|+\delta\\
    &\le \max_{i\in[n]}\left|\child_i[\hat f](\tilde x)-\child_i[\hat f]( x)\right|+\delta\\
    &\le 2\delta,
\end{align*}
where the second-to-last inequality holds thanks to our choice of $\tilde x$. 
This proves that $\cP(\tilde x)$ is feasible.

Now, let $\hat x$ denote the solution to $\cP(\tilde x)$ that is added to $\cX_\cH^{\cF}$.
We prove the correctness of the approximation by induction. 
Consider a generic degree $k\le d$ and any node/polynomial $\hat f$ of degree $k$ in the union of trees $\widehat \cF$.
If $k=1$, then $\hat f=b^\top x$ for some $b$ and we can simply consider:
\[
|b^\top x-b^\top \hat x|\le |b^\top x-b^\top \tilde x|+|b^\top \tilde x-b^\top \hat x|\\
\le \delta+2\delta = 3\delta,
\]
where the last inequality comes from our choice of $\tilde x$ (which is in the cover $\cX_{\Delta_n}^{\widehat \cF}$) and the fact that $\hat x$ is feasible for $\cP(\tilde x)$.

Let us now consider a generic degree $1\neq k\le d$ and any node $\hat f$ in $\widehat \cF$ of that degree.
As induction hypothesis, we assume that $|g(x)-g(\hat x)|\le 4(k-1)\delta$ for all polynomials $g$ of degree $k-1$ in $\widehat \cF$. Notice that we already proved that for $k-1=1$ the induction hypothesis is verified.
Then, 
\begin{align*}
    |\hat f(x)-\hat f(\hat x)| &\le |\hat f(x)-\hat f (\tilde x)|+|\hat f(\tilde x)-f(\hat x)|.
\end{align*}

The first term is easy to control. Indeed, by our choice of $\tilde x$, we have that $|\hat f(x)-\hat f(\tilde x)|\le \delta$. On the other hand, the second term requires further care.
Consider the following:
\begin{align*}
    |\hat f(\tilde x)-\hat f(\hat x)|&=\left|\sum_{i\in[n]} \tilde x_i\cdot \child_i[\hat f](\tilde x)-\hat f (\hat x)\right|\tag{\Cref{lem:decomposition}}\\
    &\le  \left|\sum_{i\in[n]} (\tilde x_i-\hat x_i)\cdot \child_i[\hat f](\tilde x)\right|+\left|\sum_{i\in[n]} \hat x_i\cdot \child_i[\hat f](\tilde x)-\hat f (\hat x)\right|\\
    &\le 2\delta+\left|\sum_{i\in[n]} \hat x_i\cdot \child_i[\hat f](\tilde x)-\hat f (\hat x)\right|\tag{Since $\hat x$ is feasible for $\cP(\tilde x)$}\\
    &\le 2\delta+\max_{i\in[n]} \left|\child_i[\hat f](\tilde x)-\child_i[\hat f](\hat x)\right|,
\end{align*}
where in the last inequality we used that $\hat f(\hat x)=\sum_{i\in[n]}\hat x_i\cdot \child_i[\hat f](\hat x)$ (by construction of the tree) together with H\"older's inequality. To control the last term, we use once again the triangle inequality:
\begin{align*}
    \max_{i\in[n]} \left|\child_i[\hat f](\tilde x)-\child_i[\hat f](\hat x)\right|&\le \max_{i\in[n]}\left|\child_i[\hat f](\tilde x)-\child_i[\hat f](x)\right|+\max_{i\in[n]}\left|\child_i[\hat f](x)-\child_i[\hat f](\hat x)\right|\\
    &\le \delta+4(k-1)\delta,
\end{align*}
where, in the last inequality, we used the properties of our choice of $\tilde x\in\cX_{\Delta_n}^{\widehat \cF}$ and the induction hypothesis (since $\child_i[\hat f]\in \widehat \cF$ and has degree $k-1$).
Thus, piecing everything together, we obtain that for every $x\in\cH$, there exists $\hat x \in \cX_{\cH}^{\cF}$ such that
\begin{align*}
    |\hat f(x)-\hat f(\hat x)|\le 4\delta+4(k-1)\delta=4k\delta.
\end{align*}

By choosing $\delta=\frac{\varepsilon}{4d}$, we can apply this to every root of the $\cF$ trees, concluding that, for every $x\in\cH$, there exists $\hat x\in \cX_{\cH}^{\cF}$ such that 
\[
|f(x)-f(\hat x)|\le \varepsilon,
\]
showing that $\cX_{\cH}^{\cF}$ is indeed a cover of the family $\cF$ over $\cH$.

To conclude the proof, we only have to analyze the running time and the size of the produced cover. It is clear that both these quantities are upper bounded by the size of the ``uncompressed'' cover $\cX_{\Delta_n}^{\widehat \cF}$, which by \Cref{cor:manyImageSimplex} is $n^{O(\rng(\widehat\cF)^2\rho(d)\log(|\widehat\cF|)/\varepsilon^2)}$. By \Cref{lem:decomposition}, we know that $\rng(\widehat\cF)$ is at most $\rho(d)\cdot \rng(\cF)$, since in $\widehat \cF$ we only placed polynomials coming from the (iterated) decomposition given by \Cref{lem:decomposition}. Moreover, we already observed that $|\widehat\cF|\le O(|\cF|\cdot n^d)$. This gives $|\cX_{\Delta_n}^{\widehat \cF}|=n^{O(\rng(\cF)^{2}\rho(d)\log(n|\cF|)/\varepsilon^2)}$, as desired.

\begin{remark}\label{rem:convex}
    It is not difficult to see that instead of polyhedral subsets of the simplex $\cH=\{x\in\Reals^n:Ax\le b,x\in\Delta_n\}$, we could also consider convex subsets $\cH=\{x\in\Reals^n:h(x)\le 0, x\in\Delta_n\}$, where $h:\Reals^n\to\Reals$ is a convex function, by replacing the linear feasibility problem \eqref{eq:feas2} in \Cref{alg:d}, with the convex feasibility problem
    \[
        \cP(\tilde x):\quad\text{find $x\in\Delta_n$ \quad s.t.\quad} 
        \begin{cases}
        h(x)\le 0\\
        \left|\sum_{i\in[n]}x_i\cdot \child_i[\hat f](\tilde x)-\hat f(\tilde x)\right|\le 2\delta\quad\forall \hat f\in\widehat{\cF}.
        \end{cases}
    \]
    However, in the case of convex sets, we necessarily have to (slightly) relax the definition of a cover, since the set $\cH$ may now only contain irrational points. 
    
    Instead, by standard results of the ellipsoid algorithm \citep{grotschel1981ellipsoid}, we can guarantee that we build a set $\cX_{\cH}^{\cF}$ such that for each $x\in\cH$, there exists $\hat x\in \cX_{\cH}^{\cF}$, so that all functions $f\in\cF$ are well approximated, but $h(\hat x)\le \beta$, and we need to spend an extra multiplicative $\log(1/\beta)$ factor in time for building the cover.
\end{remark}

\section{The General Convex Case} \label{sec:extension}

Building on top of the results in \cref{sec:many_subset}, in this section we consider the construction of $\varepsilon$-Covers for arbitrary convex compact sets. This generalization will be crucial for several of the applications considered in \cref{sec:applications}, including Constraint Satisfaction Problems and Free Games, both of which can be algebrized on sets that are more akin to hypercubes than to simplices. En route to the result, we will single out the maximum range of the polynomials on the smallest $\ell_1$ ball inscribing $\cH$ as a key driver of the complexity of the resulting covers.

To simplify the treatment without loss of generality, we consider a generic nonempty convex and compact set $\cH \subset \mathbb{R}^n_{\ge0}$ lying in the nonnegative orthant.%
We embed $\cH$ into the $(n+1)$-dimensional simplex $\Delta_{n+1}$ by introducing a dummy dimension and scaling the set by its $\ell_1$-diameter $\|\cH\|_1 \coloneqq \max_{x \in \cH} \|x\|_1$. Note that unless $\cH$ is the trivial set $\cH = \{0\}$, $\|\cH\|_1 > 0$.
Specifically, we consider the embedding $\psi: \cH \to \Delta_{n+1}$ defined as
\[
\psi(x) = \left( \frac{x_1}{\|\cH\|_1}, \dots, \frac{x_n}{\|\cH\|_1}, 1 - \frac{\|x\|_1}{\|\cH\|_1} \right).
\]

By construction, the components of $\psi(x)$ are non-negative and sum to $1$, i.e., $\psi(x)\in \Delta_{n+1}$ for each $x \in \cH$. It can be easily shown that the image $\psi(\cH)$ is a nonempty convex subset of the simplex.

Now, we cannot simply use $\rng(f)$ in our bound, since we want to recover approximations for the original optimization problem. Instead, we will use the range of the polynomial over the smallest $\ell_1$-ball that inscribes $\cH$. Formally, we let $\mathcal{B}_1(\cH)=\{x\in \Reals_{\ge0}^n: \lVert x\rVert_1\le \|\cH\|_1\}$. Then, we define the range of a function $f$ as
\[
   \widetilde{\rng}(f) = \max_{x \in \mathcal{B}_1(\cH)} f(x) - \min_{x \in \mathcal{B}_1(\cH)} f(x),
\]
where for the ease of exposition we ignore the dependency on the set $\cH$.
Moreover, we define $\widetilde{\rng}(\mathcal{F}) = \max_{f \in \mathcal{F}} \widetilde{\rng}(f)$. This leads to the following corollary:

\begin{theorem}\label{thm:manySubsetArbitrary}
    Let $\cF$ be a family of degree $d$ polynomials, and let $\cH\subset \Reals^n_{\ge0}$ be a polytope defined by $k$ inequalities. There exists an algorithm that runs in time $\poly(k)\cdot n^{O(\widetilde{\rng}(\cF)^2\rho(d)\log(n|\cF|)/\varepsilon^2)}$ that 
    returns a discrete set $\cX\subseteq \cH$ of size $n^{O(\widetilde{\rng}(\cF)^{2}\rho(d)\log(n|\cF|)/\varepsilon^2)}$ which is an $\varepsilon$-Cover of $\cF$ over $\cH$.
\end{theorem}

\begin{proof}
    We consider the set $\tilde \cH \subset \mathbb{R}^{n+1}_{\ge0}$ as the image of $\cH$ through $\psi$.
    It is easy to observe that $\tilde \cH \subseteq \Delta_{n+1}$. Moreover, since $\cH$ is a polytope and $\psi$ is affine, then $\tilde \cH$ is also a polytope. Finally, we define a family of functions $\tilde \cF$ as follows. For any $\tilde x \in \Delta_{n+1}$, let $\psi^{-1}(\tilde x) \coloneqq \|\cH\|_1 \tilde x_{1:n}$, where $\tilde x_{1:n}$ denotes the subvector of the first $n$ components of $\tilde x$. For each $f\in \cF$, we add to $\tilde \cF$ a function $\tilde f:\Delta_{n+1}\rightarrow \mathbb{R}$. In particular, we let $\tilde f(\tilde x)= f(\psi^{-1}(\tilde x))$. Under the mappings $\psi$ and $\psi^{-1}$, $\tilde \cH$ correspond one-to-one with $\cH$, and $\Delta_{n+1}$ corresponds one-to-one with $\mathcal{B}_1(\cH)$. Hence, the range $\widetilde{\rng}(f) = \rng(\tilde f)$ for all $f$, and the image of an $\varepsilon$-Cover of $\tilde \cF$ through $\psi^{-1}$ recovers an $\varepsilon$-Cover of $\cF$. 

    We can then apply \Cref{thm:manySubset} with $\tilde \cF$ as the family of functions and $\tilde \cH$ as the subset of the simplex. We obtain that there exists an algorithm that runs in time $\poly(k)\cdot n^{O(\rng(\tilde \cF)^{2}\rho(d)\log(n|\tilde\cF|)/\varepsilon^2)}$ and constructs a discrete set $\tilde\cX\subseteq \tilde \cH$ of size $n^{O(\rng(\tilde \cF)^{2}\rho(d)\log(n|\tilde \cF|)/\varepsilon^2)}$ which is an $\varepsilon$-Cover of $\tilde \cF$ over $\tilde \cH$. Taking the image of the cover under $\psi^{-1}$ yields an $\varepsilon$-Cover for $\cF$. Using the fact that $\widetilde{\rng}(\cF) = \rng(\tilde\cF)$, we obtain the guarantees of the statement.
\end{proof}

Notice that \Cref{thm:manySubsetArbitrary} requires a bounded range over the ambient $\ell_1$-ball, not just over the feasible set.
Interestingly, we cannot improve the statement of the above theorem by requiring constant range only on the feasible set, as this would contradict known hardness-of-approximation results, for instance, for sparse $3$-SAT.

As a direct corollary, we get an algorithm to approximately minimize a fixed-degree polynomial over a convex polytope. The proof is equivalent to that of \Cref{cor:min_subset}.

\begin{corollary}\label{cor:min_subsetArbitrary}
     Let $f$ be a degree $d$ polynomial, and $\cH\subset \Reals^n_{\ge0}$ be a polytope defined by $k$ inequalities. There exists an algorithm that runs in time $\poly(k)\cdot n^{O(\widetilde\rng(f)^{2}\rho(d)\log(n)/\varepsilon^2)}$ and returns an $x\in \cH$ such that $f(x)\le \min_{x\in\cH} f(x)+\varepsilon$.
\end{corollary}

\begin{remark}
    One could define an extension-complexity generalization of the above theorem whereby, for a set $\cH\subset \Reals^n$, we consider the ``best'' polyhedral set $\tilde\cH$ with $\poly(n)$ vertices that inscribes $\cH$, seeking to minimize the range of the polynomial family on $\tilde\cH$ (aiming for a constant or polylogarithmic range). In our applications, the above embedding is sufficient; we defer the study of this generalization to future work.
\end{remark}

\section{Applications}\label{sec:applications_real}

In this section, we demonstrate how \Cref{thm:manySubsetArbitrary} and \Cref{cor:min_subsetArbitrary} (and their specialized versions for subsets of the simplex) allow us to recover and extend several known $\QPTAS$es, providing a more unified perspective on these optimization problems.

\subsection{Constraint Satisfaction Problems}\label{sec:csp}

Our first application concerns $k$-Constraint Satisfaction Problems (CSPs). We consider a $k$-CSP defined over $m$ variables with an alphabet of size $q$, which we formulate as a polynomial maximization problem over a product of simplexes as follows.

We represent the $m$ variables using the set $\{x_{i,j}\}_{i \in [m], j \in [q]}$, where $x_{i,j}$ corresponds to the probability of variable $i$ taking value $j$. The domain $\cH \subseteq \mathbb{R}^{mq}_{\ge0}$ is the set of vectors satisfying $\sum_{j \in [q]} x_{i,j} = 1$ for all $i \in [m]$. Each constraint $i \in C$ involves a subset of variables $\{i_1, \dots, i_k\}$ and is defined by the polynomial
\[P_i(x) = \sum_{j_1, \dots, j_k \in [q]^k} V_i(j_1, \dots, j_k) \prod_{\ell\in[k]} x_{i_\ell, j_\ell},\]
where $V_i: [q]^k \to [0, 1]$ represents the value of the constraint as a function of the assignment.
The objective function, representing the total value of the assignment, is $f(x) = \sum_{i \in C} P_i(x)$. By applying  \Cref{cor:min_subsetArbitrary}, we obtain the following \QPTAS.

\begin{corollary}
     Given a k-CSP with a set of $m$ variables over alphabet $[q]$, there exists an algorithm that runs in time $ (mq)^{O(\rho(k)\log(mq)/\varepsilon^2)}$ and returns an $\varepsilon m^k$-approximation of the optimal assignment.
\end{corollary}

\begin{proof}
    It is sufficient to compute the quantity $\widetilde \rng(f)$. Indeed,
    it is easy to see that $\|\cH\|_1=m$, which implies: 
    \begin{align*}
    \widetilde \rng(f)&\le  \max_{x \in \mathcal{B}_1(\cH)} \sum_{i \in C} \sum_{j_1,\ldots,j_k\in[q]^k} V_i(j_1,\ldots,j_k) \prod_{\ell \in [k]} x_{i_\ell,j_\ell}\\
    &\le m^k \max_{z \in \Delta_{mq}} \sum_{i \in C} \sum_{j_1,\ldots,j_k\in[q]^k} V_i(j_1,\ldots,j_k) \prod_{\ell \in [k]} z_{i_\ell,j_\ell}\\
    & \le m^k \max_{z \in \Delta_{mq}} \sum_{i \in C} \sum_{j_1,\ldots,j_k\in[q]^k}  \prod_{\ell \in [k]} z_{i_\ell,j_\ell}\\
    & \le m^k \max_{z \in \Delta_{mq}} \sum_{i \in [m]^k} \sum_{j_1,\ldots,j_k\in[q]^k}  \prod_{\ell \in [k]} z_{i_\ell,j_\ell}\\
    & \le m^k \max_{z \in \Delta_{mq}} \prod_{\ell \in [k]} \sum_{i \in [m]} \sum_{j\in[q]}   z_{i,j}\\
    & \le m^k,
    \end{align*}

    Hence, applying \Cref{cor:min_subsetArbitrary} with $\varepsilon'=m^k \varepsilon$ we get the desired result.
\end{proof}

The previous result can then be used in conjunction with the standard greedy de-randomization to produce an optimal deterministic assignment from the randomized one found by our procedure.

A few remarks are in order. First, a \QPTAS result is the best achievable result for this problem, for non-constant alphabet sizes $q$ \citep{manurangsi2017birthday}. 
Second, as observed in \citet{arora1995polynomial}, this result is particularly useful for dense CSPs, where the number of constraints $|C|$ is proportional to $m^k$. In such cases, the additive error $\varepsilon m^k$ translates directly into a $(1+\varepsilon)$ multiplicative error.
Finally, we remark that our framework yields a suboptimal dependency on $m$ (the $\log(m)$ factor in the exponent) because it is designed for arbitrary convex sets. While specialized CSP algorithms might exploit the specific product-of-simplexes structure of $\cH$ to achieve a better running time, our approach could handle much more complex constraints on the multilinear extension of the assignment, without modification.

\subsection{Free Games}\label{sec:free}

Our second application is to Free Games \citep{aaronson2014multiple}. Although a \QPTAS for this setting was previously obtained by \citet{aaronson2014multiple}, we show that our framework easily recovers this result.

Formally, a $k$-player free game $G$ consists of:
\begin{itemize}
    \item a finite question set $Y$ and a finite answer sets $B$; and
    \item a verification function $V:  Y^k \times  B^k \rightarrow [0,1]$.
\end{itemize}

In this setting, the goal is to optimize a (randomized) strategy profile represented by a vector $x = \{x_{i,y,b}\}$ where $i \in [k]$ denotes the player, $y \in Y$ the question, and $b \in B$ the answer. The set $\cH$ represents the valid strategy space, where for each player and each question, the probabilities sum to one:
\[
\cH = \left\{ x \in \mathbb{R}_{\ge0}^{k|Y||B|} : \forall i \in [k], \forall y \in Y, \sum_{b \in B} x_{i,y,b} = 1 \right\}
\]

Intuitively, $x_{i,y,b}$ encodes a response function mapping a question to a randomized answer for each prover. Assuming, as usual, that questions are chosen independently and uniformly at random for each prover, the goal is to maximize the expected value of the game:

\[f(x) = \frac{1}{|Y|^k} \sum_{y \in Y^k} \sum_{b \in B^k} V(y, b) \prod_{i=1}^k x_{i,y_i,b_i}.\]

The factor $1/|Y|^k$ normalizes the value of the game $f(x)$ within $[0, 1]$.
Applying our framework, we recover the following guarantee.

\begin{corollary}
    Given a free game $G= (k, Y, B, V)$, there exists an algorithm that runs in time $ {n}^{O(\rho(k)\log(n)/\varepsilon^2)}$ and returns an $\varepsilon$-approximation of the optimal response function, where $n= k|Y||B|$. 
\end{corollary}

\begin{proof}
It is sufficient to compute the value of $\|\cH\|_1$ and $\widetilde\rng(f)$, and apply \Cref{cor:min_subsetArbitrary}. 
First, for any $x \in \cH$, the $\ell_1$-norm is:
\[\|x\|_1 = \sum_{i\in[k]} \sum_{y \in Y} \sum_{b \in B} x_{i,y,b} = \sum_{i\in[k]} \sum_{y \in Y} 1 = k|Y|.\]
Thus, $\|\cH\|_1 = k|Y|$.
Next, we bound $\widetilde{\rng}(f)$.  In particular, it holds:
\begin{align*}
    \widetilde \rng(f)& \le \max_{x \in \mathcal{B}_1(\cH)} \frac{1}{|Y|^k} \sum_{y\in Y^k} \sum_{b\in B^k} V(y, b) \prod_{i \in [k]} x_{i,y_i,b_i} \\
    & = k^k\max_{x \in \mathcal{B}_1(\cH)}\sum_{y\in Y^k} \sum_{b\in B^k} V(y, b) \prod_{i \in [k]} \frac{x_{i,y_i,b_i}}{|Y| k} \\
    &\le k^k  \max_{z \in \Delta_{k|Y||B|}}  \sum_{y\in Y^k} \sum_{b\in B^k} V(y, b) \prod_{i \in [k]} z_{i,y_i,b_i}\\
    & \le  k^k  \max_{z \in \Delta_{k|Y||B|}}  \sum_{y\in Y^k} \sum_{b\in B^k}  \prod_{i \in [k]} z_{i,y_i,b_i}\\
    & \le  k^k \max_{z \in \Delta_{k|Y||B|}} \prod_{i \in [k]}  \sum_{y_i\in Y^k} \sum_{b_i\in B^k}  z_{i,y_i,b_i}\\
    & \le  k^k   \prod_{i \in [k]} \frac{1}{k} \\
    & =  1,
\end{align*}
where in the second-to-last inequality we noticed that the product is maximized when each element is $\frac{1}{k}$.
Applying \Cref{cor:min_subsetArbitrary} with the computed values of $\widetilde\rng(f)$ we get the desired result.
\end{proof}
This recovers the \QPTAS of \citet{aaronson2014multiple} whenever $k$ is a fixed constant. %

\subsection{Polynomial Variational Inequalities and Local Nash Equilibria}\label{sec:VI}

In this section, we give a non-obvious application related to Variational Inequalities (VI) with polynomial operators and equilibria in nonconvex games. In particular, we show that VIs with operators $\bm{f}=[f^{(1)},\ldots,f^{(k)}]$ with $f^{(j)}\in\Reals_d[x_1,\ldots, x_k]$ over the convex hull of $n$ points can be solved in quasi-polynomial time. This result holds even when we look for optimal solutions to the VI, according to a polynomial objective function $g\in\Reals_d[x_1,\ldots,x_k]$. Formally, we consider the following problem:

\begin{definition}[Polynomial-VI (with objective function)]
    Given $n$ vectors $\{v_i\}_{i\in[n]}$ in $\Reals^k$, let $\cV= \textnormal{co}(v_1,\ldots, v_n)$. Let $\bm{f}:\Reals^k\to\Reals^k$ be a polynomial operator of degree $d$ and $g\in\Reals_d[x_1,\ldots,x_k]$ a polynomial objective function of degree $d$. Define 
    \begin{align*}
        \textsc{Opt}:=&\max_{v \in \cV} g(v)\\
        & \textnormal{s.t.}\quad \bm{f}(v)^\top (v'-v)\le 0 \quad\forall v'\in \cV.
    \end{align*}
    
    An $\varepsilon$-approximate solution is a point $v\in\cV$ such that
    \begin{enumerate}
        \item $g(v)\ge \textsc{Opt}-\varepsilon$
        \item $\bm{f}(v)^\top(v'-v)\le\varepsilon$ for all $v'\in \cV$.
    \end{enumerate}
\end{definition}

\begin{remark}
    The previous definition is well defined since the set of solutions to the VI $\{v\in \cV:\bm{f}(v)^\top(v'-v)\le 0\}$ is non-empty and compact \citep{facchinei2003finite}, and thus the problem admits a maximum.
\end{remark}

Here, it is quite natural to see the problem as an optimization problem over the simplex. Then, it is not difficult to see that \Cref{cor:manyImageSimplex} (together with \Cref{rem:polytope}) leads to a quasi-polynomial algorithm for the problem. Indeed, we can define the family of $n+1$ polynomials
\[
\cF:=\{g(Vx), \bm{f}(Vx)^\top V(e_1-x),\ldots, \bm{f}(Vx)^\top V(e_n-x)\},
\]
where $V\in\Reals^{k\times n}$ is the matrix which columns are $v_1,\ldots, v_n$, and build a $\varepsilon$-Cover of $\cF$. 
This construction readily yields the following corollary of \Cref{cor:manyImageSimplex}.

\begin{corollary}\label{th:VIs}
    There is an algorithm running in time $n^{O(\rng(\cF)^2 \log(n)\rho(d)/\varepsilon^2)}$ that finds a constant approximation of a Polynomial-VI (with objective function).
\end{corollary}
\begin{proof}
    Let $\cX^{\cF}_{\Delta_n}$ be an $\varepsilon$-Cover of $\cF$ over $\Delta_n$. Then, for every $x\in\Delta_n$, there is a $\hat x\in \cX^{\cF}_{\Delta_n}$ such that
    \[
    |g(Vx)-g(V\hat x)|\le \varepsilon,
    \quad
    \text{and}
    \quad
    |\bm{f}(Vx)^\top V(e_i-x)-\bm{f} (V\hat x)^\top V(e_i-\hat x)|\le \varepsilon.
    \]
    Now take an optimal solution $v$ to the VI, and consider the $x\in\Delta_n$ such that $Vx=v$. Moreover, we define $\hat v=V\hat x$. Thus, the first condition gives $g(\hat v)\ge \textsc{Opt}-\varepsilon$. Also, the second condition gives that $\bm{f}(\hat v)^\top(v_i-\hat v)\le\bm{f}(v)^\top(v_i-v)+\varepsilon\le \varepsilon$ for all $i\in[n]$ and, by linearity, also for all $v'\in\cV$. Moreover the algorithm runs in time $|\cX_{\Delta_n}^\cF|=n^{O(\log(|\cF|)\rho(d)\rng(\cF)^2/\varepsilon^2)}$, proving the statement.
\end{proof}

To showcase the generality of this method, we can readily give an application to the problem of computing equilibria of polynomial games over the simplex. Since determining the existence of Nash equilibria in polynomial games is \NP-hard \citep{daskalakis2021complexity}, we focus on local equilibria, defined as follows.

\begin{definition}[$\varepsilon$-Local-Nash Equilibrium]
Given $m$ players, each with strategy set $\Delta_n$, let the utility function of the $i$-th player be a polynomial $u_i: (\Delta_n)^m\rightarrow [0,1]$  of degree $d$. Find a tuple of strategies $x\in (\Delta_n)^m$ such that 
\[ 
\nabla_{x_i} u_i(x) (x'_i-x_i)\le \varepsilon \quad \forall i \in [m], x_i' \in \Delta_n.
\]
\end{definition}

Local Nash equilibrium can be mapped to a Polynomial-VI by a simple change of variables. Indeed, it is easy to see that $(\Delta_n)^m$  is a polytope with $n^m$ vertices and in which each vertex has $\ell_1$ norm of exactly $m$. Thus, \Cref{th:VIs} implies the following.

\begin{corollary}
    There exists an algorithm to find an $\varepsilon$-local-Nash equilibrium in time 
    \[ 
    n^{O(m^2\log(n)\rho(d)/\varepsilon^2)},
    \]
    where $m$ is the number of players, $n$ is the number of strategies of each player, and $d$ is the degree of the utility functions.
\end{corollary}
\begin{proof}
    We already established that the number of vertices of the polytope on which the VI is defined is $n^m$, thus a bound of $n^{O(m^2\log(n)\rho(d)R^2/\varepsilon^2)}$ is already given by \Cref{th:VIs}, where $R=\max_{i\in[m]}\rng(x_i\mapsto\nabla_{x_i} u_i(x) (e_i-x_i))$.
    We are thus left to prove that $\rng(x_i\mapsto\nabla_{x_i} u_i(x) (e_i-x_i))$ is $O(\rho(d))$. \Cref{lem:tangent} lets us show that $R\le \rho(d)\cdot\rng(u_i)\le\rho(d)$ concluding the proof.
\end{proof}

This algorithm yields a \QPTAS when the number of players and the degree are constant. This result is a generalization of the result by \citet{lipton2003playing} to non-linear utilities.

\begin{remark}
    The \QPTAS of \citet{lipton2003playing} was extended to Lipschitz games in \citet{deligkas2016lipschitz, babichenko2013small, deligkas2018approximating}. Indeed, it is not difficult to see that an $\ell_p$ grid of $\Delta_n$ (for $p\ge 2$, which can be constructed as per \citep{barman2018approximating}) and an $O(1)$-Lipschitz utility function in $\ell_p$, imply a \QPTAS by a simple Lipschitzness argument. However, as we already discussed in \Cref{sec:related}, not even linear functions are Lipschitz with respect to norms other than $\ell_1$, so these results cannot be applied in our setting.
\end{remark}

\subsection{Densest $k$-Subhypergraph}\label{sec:hyper}

A less straightforward application concerns finding dense sub-hypergraphs. In particular, we will study a generalization of the normalized densest $k$-subgraph problem (NDkS) of \citet{barman2018approximating} to hypergraphs.

Formally, let $H=(V,E)$ be a $d$-uniform hypergraph, meaning that each hyperedge $e\in E\subseteq 2^V$ has cardinality $d$.
Then, given a subset of the vertices $S\subseteq V$ of cardinality $k$, we define the (normalized) density of $S$ as $\rho(S)=\frac{|E(S)|}{k^d}$, where $E(S)=\{e\in E:e\subseteq S\}$.

\begin{definition}
In NDkSH, we are given a ($d$-uniform) hypergraph $H = (V, E)$ and a parameter $k$, and the goal is to find a set $S^*\subseteq V$ of size $k$ that has maximum density, i.e., a solution to 
\[
\max_{S\subseteq V:|S|=k} \rho(S).
\]
\end{definition}

We show how to reduce this problem to the maximization of a polynomial over a polyhedral subset of the simplex.
We start defining the generating polynomial $A_H(x)=\sum_{e\in E}\prod_{u\in e} x_u$ of the hypergraph $H$. For a $d$-uniform hypergraph, $A_H(x)$ is a degree $d$ polynomial. Then, we define the regularized function
\[
F_{\eta}(x)=A_H(x)+\eta\|x\|_2^2,
\]

and consider the following program:
\begin{subequations} \label{lp:Hyper}
    \begin{align}
    \max_{x\in \Reals^n} \ & F_\eta(x)\quad\text{s.t.}\\
    &x_i\le 1/k \quad \forall i \in V\\
    &x\in\Delta_n
    \end{align}
\end{subequations}

Before starting to prove the correctness of program \eqref{lp:Hyper} in solving NDkSH, we are going to use the following classical inequality on elementary symmetric polynomials known as Maclaurin's inequality (see, e.g., \citep[][]{pecaric2005generalization}).

\begin{lemma}[Maclaurin's inequality]\label{lem:maclaurin}
    Let $\mathcal{E}_s(x)=\sum_{S\subset [n]: |S|=s}\prod_{k\in S}x_k$ be the degree $s$ elementary symmetric polynomial in $n$ variables. Then,
    \[
    \left(\frac{\mathcal{E}_n(x)}{\binom{n}{n}}\right)^{1/n}\le\left(\frac{\mathcal{E}_{n-1}(x)}{\binom{n}{n-1}}\right)^{1/(n-1)} \le \cdots\le\frac{\mathcal{E}_1(x)}{\binom{n}{1}} = \frac{\sum_{i\in[n]}x_i}{n}.
    \]
\end{lemma}

Then, we show how, given a solution to program \eqref{lp:Hyper}, it is possible to obtain an ``integer'' solution. Formally:

\begin{lemma}[Rounding lemma]\label{lem:rounding}
    Given a feasible solution $x\in \Delta_n$ to program \eqref{lp:Hyper}, if $\eta\ge\frac{1}{2}$, there is a polynomial-time algorithm that returns an $y\in \{0,1/k\}^n\cap \Delta_n$ such that $F_\eta(y)\ge F_\eta(x)$.
\end{lemma}

\begin{proof}
    Given a solution $x\in [0,1/k]^n\cap\Delta_n $, we iteratively reduce the number of fractional vertices without decreasing the value of the function.
    Consider the set of fractional vertices $M(x)=\{u\in V, x_u\in(0,1/k)\}$ and let
    \begin{align*}
        &i\in\arg\max_{k\in M(x)} \partial_{k}F_\eta(x)\quad\text{and}\quad j\in\arg\min_{k\in M(x)} \partial_{k}F_\eta(x),
    \end{align*}
    where $\partial_{k}F_\eta(x)$ denotes the partial derivative with respect to $x_k$.
    
    Consider the function $\phi(t)=F_\eta(x+t(e_i-e_j))$ and  let $\delta=\min(1/k-x_i,x_j)$. Our goal will be to show that $\phi(t)$ is increasing and hence setting $x'=x+\delta(e_i-e_j)$ will not decrease the value.

    To do that, we start computing the first-order and second-order derivatives. In particular:
    \begin{align*}
    &\phi'(t)=\nabla F_\eta(x+t(e_i-e_j))^\top(e_i-e_j)=\partial_iF_\eta(x+t(e_i-e_j))-\partial_jF_\eta(x+t(e_i-e_j))\\
    & \phi''(t)=\partial^2_{ii}F_\eta(x+t(e_i-e_j))+\partial^2_{jj}F_\eta(x+t(e_i-e_j))-2\partial^2_{ij}F_\eta(x+t(e_i-e_j))
    \end{align*}
    
    It is easy to verify that 
    \begin{itemize}
        \item $\partial_{i} F_\eta(x)=2\eta x_i+\sum_{e\in E:i\in e} \prod_{j\in e: j\neq i}x_j$;
        \item $\partial_{ii}F_\eta(x)=2\eta$;
        \item $\partial_{ij}F_\eta(x)=\sum_{e\in E: i,j\in E}\prod_{k\in e: k\neq i,j}x_k$.
    \end{itemize}

    Now, we proceed to upperbound $\partial_{ij}F_\eta(x)$ as follows:

\begin{align}\label{eq:elementary}
    \partial_{ij}F_\eta(x)&=\sum_{\substack{e\in E\\\{i,j\}\subseteq e}}\prod_{k\neq i,j} x_k \le \sum_{\substack{S\subseteq[n]\\|S|=d-2}}\prod_{k\in S} x_k.
\end{align}

We recall the definition of $\mathcal{E}_s(x)=\sum_{S\subset [n]: |S|=s}\prod_{k\in S}x_k$ to be the degree $s$ elementary symmetric polynomial in $n$ variables given in \Cref{lem:maclaurin}. Which implies
\[
 \partial_{ij}F_\eta(x)\le \mathcal{E}_{d-2}(x),
\]
then, by applying \Cref{lem:maclaurin} to \Cref{eq:elementary} we obtain that:
\[
\left({\frac{\mathcal{E}_{d-2}(x)}{\binom{n}{d-2}}}\right)^{\frac{1}{d-2}}\le\frac{\mathcal{E}_1(x)}{\binom{n}{1}}.
\]
Rearranging, we get:
\[
\partial_{ij}F_\eta(x)\le \mathcal{E}_{d-2}(x)\le\left(\frac{\mathcal{E}_1(x)}{\binom{n}{1}}\right)^{d-2}\binom{n}{d-2}=\left(\frac{\sum_{k\in [n]}x_k}{n}\right)^{d-2}\binom{n}{d-2}=\frac{1}{n^{d-2}}\binom{n}{d-2}\le 1.
\]

Now, we can conclude that $\phi''(t)\ge 0$ for all $t \in [0,\delta] $.
Indeed, it holds:
    \begin{align*}
    \phi''(t)&=4\eta-2\sum_{e\in E: i,j\in E}\,\prod_{k\in e: k\neq i,j}(x+t(e_i-e_j))_k\\
    &=4\eta-2\sum_{e\in E, i,j\in E}\,\prod_{k\in e, k\neq i,j}x_k\\
    & \ge 0,
\end{align*}
where the last inequality holds setting $\eta\ge \frac{1}{2}$. 
Moreover, we observe that $\phi'(0)=\partial_iF_\eta(x)-\partial_j F_\eta(x)\ge 0$ by our choice of $i$ and $j$.

Now it is clear that $\phi$ is monotonically increasing in $[0,\delta]$. Indeed, $\phi'(t)\ge 0$ and $\phi''(t)\ge 0$ for each $t\in [0,\delta]$, \emph{i.e.}, it's derivative is increasing in $t$.

Thus, $\phi(\delta)\ge \phi(0)$ which implies that
    \[
    F_\eta(x')\ge F_\eta(x)
    \]
    where $x'=x+\delta(e_i-e_j)$. 
    
As promised, we obtained a new solution $x'$ such that $|M(x')|<|M(x)|$ as either $x'_j=x_j-\delta=0$ or $x'_i=x_i+\delta=1/k$. Repeating the procedure recursively, we eventually obtain $|M(x)|=0$.

Finally, we observe that we can apply the procedure only if $|M(x)|\ge 2$. Hence, we need also to prove that $|M(x)|=0$ or $|M(x)|\ge 2$. This holds since $x\in\Delta_n$. Indeed, if $|M(x)|=1$ then we would have that all but one $x_i$ are either $0$ or $1/k$ (call $c\in\Naturals$ the number of $x_i=1/k$) and one $x_i$ is equal to $\alpha/k$ for some $\alpha\in(0,1)$. Then $\sum_ix_i=c/k+\alpha/k=1$ but this is impossible since $c,k\in\Naturals$ and $\alpha\in (0,1)$.
\end{proof} 

We proceed by showing that for any solution $x$ to program \eqref{lp:Hyper}, after being rounded by applying the algorithm detailed in the proof of \Cref{lem:rounding}, the value of $F_\eta(x)$ is a proxy of the density $\rho(\mathrm{supp}(x))$, where we define $\mathrm{supp}(x)$ as the set of vertices $v$ with a corresponding variable $x_v>0$. 

\begin{lemma}\label{lem:pure}
	For any $x\in\Delta_n$ such that $x_i\in\{0,1/k\}$,
    we have that $|\mathrm{supp}(x)|=k$ and  $F_\eta(x)=\frac{\eta}{k}+\rho(\mathrm{supp}(x))$.
\end{lemma}
\begin{proof}
	It is clear that $|\mathrm{supp}(x)|=k$. Moreover, it is easy to see that the two components of $F_\eta(x)$, \emph{i.e.}, $A_H(x)$ and $\eta \lVert x\rVert^2_2$ have value
	\begin{align*}
		 &A_H(x)=\sum_{e\in E}\frac{1}{k^d}\prod_{i\in e}\mathbb{I}(i\in \mathrm{supp}(x))=\sum_{e\in E}\frac{1}{k^d}\mathbb{I}(e\subseteq \mathrm{supp}(x))=\rho(\mathrm{supp}(x))
	\end{align*}
    and
    \begin{align*}
         \eta\|x\|_2^2&=\eta\sum_{i\in[n]}x_i^2= \eta\frac{|\mathrm{supp}(x)|}{k^2}=\frac\eta k,
    \end{align*}
    concluding the proof.
\end{proof}

Hence, we showed that solving program \eqref{lp:Hyper} enables us to solve NDkSH. To apply our framework derived in \Cref{sec:many_subset} (and in particular \Cref{cor:min_subset}) and provide an additive approximation, we are left to bound the range of the objective function.

\begin{lemma}\label{lem:rangeF}
It holds that
\[
\rng (F_\eta) = O(\eta).
\]
\end{lemma}

\begin{proof}
The range of $\|x\|^2_2$ is at most $1$, hence in the following we focus on the range of $A_H(x)$.

For the minimum value, observe that $A_H(e_i)=0$ is always positive.
For the maximum, we can use Maclaurin's inequality. First observe that $A_H(x)$ is always smaller than the symmetric degree $d$ polynomial, since to obtain such a polynomial we are adding hyperedges, \ie,
\[
A_H(x)\le \sum_{S\subseteq V, |S|=d}\prod_{i\in S}x_i=\mathcal{E}_d(x).
\]
Then \Cref{lem:maclaurin} gives that $\left({\frac{\mathcal{E}_d(x)}{\binom{n}{d}}}\right)^{1/d}\le \frac{\mathcal{E}_1(x)}{\binom{n}{1}}$, or equivalently:
\[
A_H(x)\le \mathcal{E}_d(x)\le\binom{n}{d}\left(\frac{\mathcal{E}_1(x)}{\binom{n}{1}}\right)^d=\binom{n}{d}\left(\frac{\sum_{i\in [n]}x_i}{n}\right)^d\le\binom{n}{d}\frac{1}{n^d}\le 1.
\]

Thus the range of $F_\eta(x)$ is controlled by $1+\eta$.
\end{proof}

Now we can prove the final theorem regarding constant additive approximations to NDkSH.

\begin{theorem}
    Let $H=(V,E)$ be a d-uniform hypergraph with $n$ vertices. Then, an $\varepsilon$-additive approximation of NDkSH can be found in time
    \[n^{\tilde O(\log(n)\rho(d)/\varepsilon^2)}.\]
\end{theorem}

\begin{proof}
From \Cref{lem:rangeF} we know that $\rng (F_\eta)=O(\eta)$ and it suffices to take $\eta\ge1/2$ to make \Cref{lem:rounding} to work. Thus, by \Cref{cor:min_subset} we can find an $\varepsilon$-approximate solution to the program above in time $n^{\tilde O(\log(n)\rho(d)/\varepsilon^2)}$.
By \Cref{lem:rounding}, we can build a solution in which $x_v\in \{0,1/k\}$ for each $v$ with at least the same value. By \Cref{lem:pure}, $\rho(\mathrm{supp}(x))$ is an $\varepsilon$ additive approximation of NDkSH.
\end{proof}

\section{Open Problems}

Our framework, while being very general, leaves open several interesting directions.

First, while we demonstrated that our approach allows for the optimization of arbitrary polynomials over subsets of the simplex, this result does not readily extend to related problems such as VIs. Specifically, our reduction from a VI to polynomial optimization relies on constructing a dedicated polynomial for each vertex, an approach that becomes intractable for polyhedral sets represented implicitly by a polynomial number of inequalities. It is unclear whether it is possible to design algorithms that run in quasi-polynomial time with respect to the number of inequalities. This would, for instance, yield positive results for constrained min-max optimization over the simplex~\citep{bernasconi2024role,anagnostides2025complexity,anagnostides2026computational,bernasconi2026complexity}.

Moreover, while our $\QPTAS$es are tight, it would be interesting to identify specific settings where more efficient algorithms, e.g., $\PTAS$es, exist. Partial results in this direction were previously obtained by \citet{barman2018approximating} for quadratic polynomials, with applications to Nash equilibria and the $k$-densest subgraph problem. Investigating whether similar results can be achieved for restricted classes of higher-degree polynomials is a promising direction, requiring us to identify the appropriate parameters (of the polynomials) that yield polynomial-sized covers.

We note that the best-known rates for many of our applications are $n^{O(\log(n)/\varepsilon^2)}$, and we were able to match the dependence on $\varepsilon$ in our results. However, the lower bounds are only of $n^{\widetilde\Omega(\log(n)/\varepsilon)}$ (see, for instance, \citet[open question~$2$]{aaronson2014multiple}, but the same gap exists for other problems such as optimal Nash and CSPs). Understanding this gap better is an interesting problem. With our techniques, the $\varepsilon^2$ comes from the concentration inequalities used in \Cref{th:probabilisticPTAS}, and it is not clear how to improve it by using a better one.

Lastly, in this paper, we placed particular emphasis on the unifying character and generalization abilities of our framework. However, this led to an exponential dependency of the sample size on the degree of the polynomial. In many applications of interest, however, we only need to consider multi-linear polynomials (such as Free Games and CSPs), and in that case, we believe that a better analysis (in particular of the bias term) could improve our algorithms to a polynomial dependence on the degree.

\appendix
\section*{Appendix}

\section{Proof of \Cref{th:probabilisticPTAS}}\label{app:prob}

Our proof builds on the work of \citet{de2015alternative}.
In a more modern take on the problem, \citet{de2015alternative} exploited the connection between Bernstein approximation $B_N[f](x)$ and minimization of polynomials over the simplex. More specifically, the order $N$ of the Bernstein approximation coincides with the number of samples $\{0,1\}^n\ni X_k \!\sim\! \mathrm{Categorical}(x)$ we draw from the categorical distribution with parameter $x\in\Delta_n$ and the Bernstein approximation $B_N[f](x)$ of $f$ at $x$ is the expected value of $f$ over the sampling procedure.
Namely, $B_N[f](x)=\mathbb{E}[f(Z)]$, where $Z=\frac1N\sum_{k=1}^N X_k$, and $X_1,\ldots, X_N$ are i.i.d.~samples of a categorical distribution with parameter $x$. %

Thus, the classic guarantees of uniform approximation for the Bernstein interpolant $B_N$ (\emph{e.g.}, \citep[Theorem~2.11]{de2006ptas}, \citep[Theorem~8]{de2015alternative}) can be interpreted as bounding the bias of the estimator $f(Z)$ under the sampling.
\begin{theorem}[{\cite[Theorem~2.11]{de2006ptas}, \cite[Theorem~8]{de2015alternative}}]\label{th:bias}
For any $x\in\Delta_n$, let $Z=\sum_{k=1}^N X_k/N$, where $X_k\stackrel{i.i.d.}{\sim} \mathrm{Categorical}(x)$. Then for any $f\in \Reals_{d}[x_1,\ldots,x_n]$ it holds:
    \[
    |\mathbb{E}[f(Z)]-f(x)| \le \frac{1}{N}\rho(d)\cdot\rng(f).
    \]
\end{theorem}

With this theorem, \citet{de2015alternative} directly derive an approximate minimizer of $f$ over $\Delta_n$. Indeed, we can observe that in \Cref{eq:Bernstein}, for each $x$, the Bernstein approximation $B_N[f](x)$ is a weighted sum of coefficients $\{f(\frac{\beta}{N})\}_{\beta\in\Naturals_N^n}$ with weights $\{\frac{N!}{\beta!}x^\beta\}_{\beta\in\Naturals_N^n}$.%
\footnote{$\{\frac{N!}{\beta!}x^\beta\}_{\beta\in\Naturals_N^n}$ are weights thanks to the multinomial theorem, as $1=(\sum_{i\in[n]}x_i)^N=\sum_{\beta\in \Naturals_{n}^N}\frac{N!}{\beta!}x^\beta$.}
Indeed, for any $x\in\Delta_n$:
\[
\min_{\beta\in\Naturals_N^n}f\left(\frac{\beta}{N}\right)-\min_{x\in\Delta_n}f(x)\le \min_{x\in\Delta_n}B_N[f](x)-\min_{x\in\Delta_n}f(x)\le \frac{1}{N}\rho(d)\cdot\rng(f),
\]
and enumerating over all the possible $\beta\in\Naturals_N^n$ takes at most $O(n^N)$-time, resulting in a \PTAS by choosing $N=\Theta(1/\varepsilon)$.
We build on this approach, combining it with a concentration argument,%
that will allow us to extract an approximate cover of $f$ rather than only an approximate minimum.

First, we prove the following technical lemma that states that any polynomial on the simplex is Lipschitz in $\ell_1$ norm on the simplex, with a Lipschitz constant that depends only on the degree and the range (and hence that is independent of the dimension).
Proving that the range controls the Lipschitzness in $\ell_1$ requires special care and non-trivial properties of the Bernstein representation. Indeed, the proof does not simply rely on showing that $\|\nabla f\|_\infty$ is only a function of the degree. For instance, take $f(x)=1-\left(\sum_{i\in[n]}x_i\right)^M$. This polynomial has range zero on $\Delta_n$, but $\|\nabla f\|_\infty=\Theta(M)$, and thus can be arbitrary large. The crucial point is to consider only directional derivatives on the tangent space of the simplex, as claimed in the following lemma which utilizes the relationship between the coefficients of the Bernstein basis and the range of the polynomial (\Cref{lem:boundBerncoeff}).%

\begin{restatable}{lemma}{tangent}\label{lem:tangent}
    Let $f\in\Reals_d[x_1,\ldots,x_n]$. For every $x\in\Delta_n$ and every $w\in\Reals^n$ satisfying $\sum_{i\in[n]}w_i=0$, we have
    \[
    |\nabla f(x)^\top w|\le \rho(d)\cdot\rng(f) \|w\|_1.
    \]
\end{restatable}

\begin{proof}
    We write the polynomial $f=\sum_{\beta\in\Naturals^n_d} c_\beta \cB^d_{\beta}(x)$ according to its Bernstein basis and take the derivative with respect to $x_i$, we obtain
    \begin{align*}
    \partial_{x_i} f(x)&=\sum_{\beta\in\Naturals^n_d:\beta_i\ge 1}c_\beta\beta_i \frac{d!}{\beta!}x^{\beta-e_i}\\
    &=\sum_{\beta\in\Naturals^n_{d-1}} c_{\beta+e_i}(\beta+e_i)_i\frac{d!}{(\beta+e_i)!}x^\beta\\
    &=\sum_{\beta\in\Naturals^n_{d-1}} c_{\beta+e_i}\frac{d!}{\beta!}x^\beta.
    \end{align*}
    Now, we can compute 
    \begin{align}\label{eq:tmp7}
    \nabla f(x)^\top w=\sum_{\beta\in\Naturals^n_{d-1}} \left(\sum_{i\in[n]}w_ic_{\beta+e_i}\right)\frac{d!}{\beta!}x^\beta.
    \end{align}
    For each $\beta\in \Naturals^n_{d-1}$ we consider
    \begin{align*}
        \sum_{i\in[n]}w_ic_{\beta+e_i}&=\sum_{i\in[n]:w_i\ge 0}w_ic_{\beta+e_i}+\sum_{i\in[n]:w_i< 0}w_ic_{\beta+e_i}\\
        &\le \left(\sum_{i\in[n]:w_i\ge 0}w_i\right)\max_{i\in[n]} c_{\beta+e_i}+\left(\sum_{i\in[n]:w_i< 0}w_i\right)\min_{i\in[n]}c_{\beta+e_i}.
    \end{align*}
    Since $\sum_{i\in[n]}w_i=0$, we have that $\|w\|_1=\sum_{i\in[n]:w_i\ge 0}w_i-\sum_{i\in[n]:w_i<0}w_i=2\sum_{i\in[n]:w_i\ge 0}w_i=-2\sum_{i\in[n]:w_i< 0}w_i$ and thus, continuing from above, we have:
    \begin{align*}
        \sum_{i\in[n]}w_ic_{\beta+e_i}&\le\frac{\|w\|_1}{2}\left(\max_{i\in[n]} c_{\beta+e_i}-\min_{i\in[n]}c_{\beta+e_i}\right)
        \le \frac{\|w\|_1}{2}\rho(d)\cdot\rng(f),
    \end{align*}
    where the last inequality comes from \Cref{lem:boundBerncoeff}.
    Plugging it back to \Cref{eq:tmp7}, and using that $\sum_{\beta\in\Naturals_{d-1}^n}\frac{(d-1)!}{\beta!}x^\beta=1$ we get:
    \begin{align*}
       \nabla f(x)^\top w\le d\frac{\|w\|_1}{2}\rho(d)\cdot\rng(f), 
    \end{align*}
    which concludes the statement, since all the same calculations can be carried on $-w$.
\end{proof}

Then, the $\ell_1$-Lipschitzness is a simple corollary:

\begin{restatable}{lemma}{lemLoneLip}\label{lem:L1Lip}
    For any $x,y\in\Delta_n$ and any $f\in\Reals_d[x_1,\ldots,x_n]$ it holds 
    \[
    |f(x)-f(y)|\le\rho(d)\cdot\rng(f)\|x-y\|_1.
    \]
\end{restatable}

\begin{proof}
To prove the statement, we can simply apply \Cref{lem:tangent}. Indeed, for any $x,y\in\Delta_n$, observe that $x-y\in\{v\in\Reals^n:\sum_{i\in[n]}v_i=0\}$ and thus, by the mean value theorem
\begin{align*}
    |f(x)-f(y)|&= |\nabla f(z)^\top (x-y)|\tag{For some $z$ in the convex combination of $x$ and $y$}\\
    &\le \rho(d)\cdot\rng(f)\|x-y\|_1\tag{\Cref{lem:tangent}},
\end{align*}
concluding the proof.
\end{proof}

Combining this lemma with simple concentration arguments, and accounting for the bias given by \Cref{th:bias}, we can easily prove that we can uniformly approximate \emph{any} $x\in \Delta_n$, i.e., generate a cover for $f$ over $\Delta_n$.

\theoremEasy*

\begin{proof}
    We will prove the concentration using McDiarmid's inequality. Indeed, define $Z=\frac{1}{N}\sum_{k\in[N]} X_k$ and $Z'=\frac{1}{N}\sum_{k\in[N], k\neq k'}X_k+\frac{1}{N}X_{k'}$ for a generic $k'\in [N]$. Then, if $|f(Z)-f(Z')|\le G$, McDiarmid's inequality gives that
    \begin{align*}
        \mathbb{P}(|\mathbb{E}[f(Z)]-f(Z)|\ge \varepsilon)\le 2\exp\left(-\frac{2\varepsilon^2}{G^2N}\right).
    \end{align*}
    Accounting for the bias of $\mathbb{E}[f(Z)]$, we obtain
    \begin{align}
        \mathbb{P}(|f(x)-f(Z)|\ge \varepsilon)&\le\mathbb{P}(|f(x)-\mathbb{E}[f(Z)]|+|f(Z)-\mathbb{E}[f(Z)]|\ge \varepsilon) \nonumber\\
        &\le \mathbb{P}\left(|f(Z)-\mathbb{E}[f(Z)]|\ge \varepsilon-\frac{1}{N}h(d)\rng(f)\right)\tag{\Cref{th:bias}} \nonumber\\
        &\le  \mathbb{P}\left(|f(Z)-\mathbb{E}[f(Z)]|\ge \varepsilon/2\right)\tag{$N\ge \frac{2}{\varepsilon}\rho(d)\rng(f)$} \nonumber\\
        &\le 2\exp\left(\frac{-\varepsilon^2}{2G^2N}\right).\label{eq:mc}
    \end{align}
    Now, we are left to bound $G$. In general, we do not expect the gradient of polynomials to be bounded on the simplex; however, we can perform a careful analysis that only considers tangential movements. 
   Observe that $Z-Z'\le\frac{1}{N}(e_i-e_j)$ for some $i,j\in[n]$ and thus $\|Z-Z'\|_1=\frac{2}{N}$.
    Hence, by \Cref{lem:L1Lip}, we can conclude that:
    \begin{align}
    |f(Z)-f(Z')|&\le G:=\frac{1}{N}\rho(d)\cdot\rng(f). \label{eq:lip}
    \end{align}
     Finally, combining \Cref{eq:mc} and \Cref{eq:lip}, we get  
     \[
     \mathbb{P}(|f(x)-f(Z)|\ge \varepsilon)\le 2\exp\left(-\frac{N\varepsilon^2 }{\rho(d)\rng(f)^2}\right),
     \]
     concluding the proof.
\end{proof}

\section{Non-Lipschitzness of Polynomials Over the Simplex} \label{app:Lip}

\begin{proposition}\label{prop:Lip}
   For any $p \ge 1$, there exists a polynomial $f\in\Reals_{1}[x]$ of degree $d=1$ and range $1$, such that $\sup_{x\neq y\in\Delta_{n}}\frac{|f(x)-f(y)|}{\|x-y\|_p}\ge \Omega(n^{(p-1)/p})$.
   In particular, this shows that bounded-range low-degree polynomials cannot be $O(1)$-Lipschitz continuous with respect to any norm other than the $\ell_1$ norm.
\end{proposition}

\begin{proof}
    Consider the linear function
    \[
        f(x) = \sum_{i \le n/2} x_i,
    \]
    and the two distributions $x$ and $y$ that respectively uniformly place mass on the first or second half of the dimensions,
    \[
        x_i = \begin{cases}\frac{2}{n} & \text{if } i \le n/2,\\
        0 & \text{otherwise}
        \end{cases}
        ,\hspace{2cm} 
        y_i = \begin{cases}0 & \text{if } i \le n/2,\\
        \frac{2}{n} & \text{otherwise}.
        \end{cases}
    \]
    Then,
    \[
        |f(x) - f(y)| = 1,
        \qquad
        \|x-y\|_p = \left( n \left(\frac{2}{n}\right)^p\right)^{1/p} = 2n^{(1-p)/p}.
    \]
    So, the Lipschitz constant is at least
    $
        \Omega(n^{(p-1)/p})
    $ as claimed.
\end{proof}

\printbibliography

\end{document}